\RequirePackage[l2tabu, orthodox]{nag} 

\newif\ifArxiv
\Arxivtrue

\newif\ifPreprint
\Preprinttrue

\ifPreprint
\documentclass[journal,onecolumn]{IEEEtran}
\else
\documentclass[journal]{IEEEtran}
\fi

\usepackage[english]{babel}  
\usepackage[T2A,T1]{fontenc} 
\usepackage{amsmath,amssymb} 
\usepackage{graphicx}        
\usepackage{color}           
\usepackage{varioref}        
\usepackage{xspace}          
\usepackage{url}             
\usepackage{cite}            
\usepackage{multirow}        
\usepackage{booktabs} 
\usepackage{stmaryrd} 
\usepackage{algorithm}
\usepackage{algpseudocode}
\usepackage{microtype} 

\usepackage[caption=false]{subfig} 
\usepackage{bm}
\usepackage[utf8]{inputenc}

\usepackage[hyperref,thref]{ntheorem}
\ifArxiv\else
\theorempreskip{0em}
\theorempostskip{0em}
\theoremheaderfont{\hspace{\parindent}\normalfont\itshape}
\fi
\theoremseparator{:}
\theorembodyfont{\normalfont\upshape}
\newtheorem{theorem}{Theorem}
\newtheorem{lemma}[theorem]{Lemma}
\newtheorem{corollary}[theorem]{Corollary}
\newtheorem{proposition}[theorem]{Proposition}

\newtheorem{definition}[theorem]{Definition}

\newtheorem{problem}[theorem]{Problem}

\newtheorem{remark}[theorem]{Remark}

\usepackage{fancyref}

\makeatletter
\def\mkfancyprefix#1#2{%
\expandafter\def\csname fancyref#1labelprefix\endcsname{#1}%
\begingroup\def\x{\endgroup\frefformat{plain}}%
    \expandafter\x\csname fancyref#1labelprefix\endcsname
    {\MakeLowercase{#2}\fancyrefdefaultspacing##1}%
\begingroup\def\x{\endgroup\Frefformat{plain}}%
    \expandafter\x\csname fancyref#1labelprefix\endcsname
    {#2\fancyrefdefaultspacing##1}%
\begingroup\def\x{\endgroup\frefformat{vario}}%
    \expandafter\x\csname fancyref#1labelprefix\endcsname
    {\MakeLowercase{#2}\fancyrefdefaultspacing##1##3}%
\begingroup\def\x{\endgroup\Frefformat{vario}}%
    \expandafter\x\csname fancyref#1labelprefix\endcsname
    {#2\fancyrefdefaultspacing##1##3}%
}
\makeatother
\fancyrefchangeprefix{\fancyrefeqlabelprefix}{eqn}
\Frefformat{vario}{\fancyrefeqlabelprefix}{(#1)#3}

\mkfancyprefix{ssec}{Section}
\mkfancyprefix{tbl}{Table}
\mkfancyprefix{thm}{Theorem}
\mkfancyprefix{lem}{Lemma}
\mkfancyprefix{cor}{Corollary}
\mkfancyprefix{prop}{Proposition}
\mkfancyprefix{prob}{Problem}
\mkfancyprefix{alg}{Algorithm}
\mkfancyprefix{inv}{Invariant}
\mkfancyprefix{ex}{Example}
\mkfancyprefix{line}{Line}
\mkfancyprefix{def}{Definition}
\mkfancyprefix{itm}{Item}
\mkfancyprefix{app}{Appendix}
\newcommand{\cref}[1]{\Fref{#1}}
\renewcommand{\vref}[1]{\Fref[vario]{#1}}

\AtBeginDocument{}

\newcommand\ifnull[3]{%
  \ifx\null#1%
    #2%
  \else%
    #3%
  \fi}

\newcommand{\Grobner}{Gr\"o{}bner\xspace} 

\newcommand{\word}[1]{\textnormal{#1}}

\makeatletter
\def\easycyrsymbol#1{\mathord{\mathchoice
  {\mbox{\fontsize\tf@size\z@\usefont{T2A}{cmr}{m}{n}#1}}
  {\mbox{\fontsize\tf@size\z@\usefont{T2A}{cmr}{m}{n}#1}}
  {\mbox{\fontsize\sf@size\z@\usefont{T2A}{cmr}{m}{n}#1}}
  {\mbox{\fontsize\ssf@size\z@\usefont{T2A}{cmr}{m}{n}#1}}
}}
\makeatother
\newcommand{\Ya}{\easycyrsymbol{\CYRYA}}

\newcommand\half{\tfrac 1 2}
\newcommand\defeq{\triangleq}            
\newcommand\modop{\ \word{mod}\ }         
\newcommand{\mo}{{-1}}                   
\newcommand{\T}[1]{^{(#1)}}              
\renewcommand\vec[1]{\bm{#1}}            
\newcommand\mtrx[1]{\begin{pmatrix}#1\end{pmatrix}} 
\newcommand{\LM}{\textnormal{\footnotesize LM}}
\newcommand{\LC}{\textnormal{\footnotesize LC}}

\newcommand{\diag}{\word{diag}}

\newcommand{\Mod}[1]{\mathcal{#1}} 
\newcommand{\dimsq}[1]{#1 \times #1}
\newcommand{\floor}[1]{\lfloor #1 \rfloor}
\newcommand{\ceil}[1]{\lceil #1 \rceil}
\newcommand{\rowdeg}[1]{\word{rowdeg}\,#1} 
\newcommand{\weight}[1]{\word{weight}(#1)}  

\newcommand\F{\mathbb F\xspace} 
\newcommand\FF[1]{\mathbb F_{#1}\xspace} 
\newcommand\Fqq{\FF{q^2}}

\newcommand\RR{\mathbb R\xspace}
\newcommand\ZZ{\mathbb Z\xspace}
\newcommand\NN{\mathbb N\xspace}
\newcommand{\Csel}[2]{#1_{[#2]}} 

\newcommand{\LP}{\textnormal{\footnotesize LP}}
\newcommand{\OD}[1]{{\Delta(#1)}} 
\newcommand{\Transp}{^\top}
\newcommand{\dG}{d^\star}
\newcommand{\PowerSeriesRing}[2][x]{#2\llbracket #1 \rrbracket} 

\newcommand{\WmapO}[1]{\Phi_{#1}}
\newcommand{\Wmap}{\WmapO{\nu, \vec w}}
\newcommand{\WmapH}{\WmapO{q, \vec w}}
\newcommand{\PmapO}[1]{\Psi_{#1}}
\newcommand{\Pmap}{\PmapO{\nu, \vec w}}
\newcommand{\PmapH}{\PmapO{q, \vec w}}

\newcommand{\params}[3]{[#1,\ #2,\ #3]}
\newcommand{\Code}{\mathcal C}
\newcommand\Errs{\mathcal E} 

\newcommand{\sell}{{s,\ell}}
\newcommand{\tauPow}{\tau_{\word{Pow}}}
\newcommand{\tauGS}{\tau_{\word{GS}}}

\newcommand\Herm{\mathcal H}
\newcommand\Ring{\Ya}
\newcommand\Places{\mathcal P}
\newcommand\PlacesAff{\mathcal P^\star}
\renewcommand{\L}{\mathcal L}
\newcommand\dimL{\word{dim}\ \mathcal L}
\newcommand\order{\deg_{\Herm}}
\newcommand\orderz[1]{\deg_{\Herm,#1}}
\newcommand\vectify[1]{\ifnull{#1}{{\curlyvee}}{{\curlyvee}(#1)}}
\newcommand\vectifyinv[1]{\ifnull{#1}{{\curlyvee^\mo}}{{\curlyvee^\mo}(#1)}}
\newcommand\vectifyz[1]{\ifnull{#1}{{\curlyvee_{\!\!z}}}{{\curlyvee_{\!\!z}}(#1)}}
\newcommand\matmult[1]{{\amalg}_{#1}}
\newcommand\Div[1]{\word{div}(#1)}

\newcommand{\costPoly}[1]{\word{P}(#1)}

\newcommand\Oapp{O^{\scriptscriptstyle \sim}\!}

\algrenewcommand\alglinenumber[1]{{\scriptsize#1}}   
\algrenewcommand\algorithmicrequire{\textbf{Input:}} 
\algrenewcommand\algorithmicensure{\textbf{Output:}} 
\newcommand{\Ifline}[2]{\State \textbf{if }#1{ \textbf{then} }#2} 
\newcommand{\ass}{\leftarrow}

\begin{document}

\title{Sub-quadratic Decoding of One-point Hermitian Codes}

\author{Johan~S.~R.~Nielsen,
        Peter~Beelen%
        \thanks{%
          J.~S.~R.~Nielsen is with the GRACE Project, INRIA Saclay \& LIX, \'Ecole Polytechnique, France (e-mail: jsrn@jsrn.dk).
          P.~Beelen is with the Department of Applied Mathematics and Computer Science, Technical University of Denmark (e-mail: pabe@dtu.dk).
          \ifPreprint
          This is a preprint of work that has been published in IEEE Transactions of Information Theory.
          DOI: 10.1109/TIT.2015.2424415.
          \fi
          Copyright \copyright~2015 IEEE. Personal use of this material is permitted.  However, permission to use this material for any other purposes must be obtained from the IEEE by sending a request to pubs-permissions@ieee.org.
        }
}

\maketitle

\begin{abstract}
  We present the first two sub-quadratic complexity decoding algorithms for one-point Hermitian codes.
  The first is based on a fast realisation of the Guruswami--Sudan algorithm by using state-of-the-art algorithms from computer algebra for polynomial-ring matrix minimisation.
  The second is a Power decoding algorithm: an extension of classical key equation decoding which gives a probabilistic decoding algorithm up to the Sudan radius.
  We show how the resulting key equations can be solved by the matrix minimisation algorithms from computer algebra, yielding similar asymptotic complexities.
\end{abstract}


\section{Introduction}

\IEEEPARstart{I}{n} this article we examine fast decoding of one-point Hermitian codes beyond half the minimum distance.
First we give a new algorithm for constructing the interpolation polynomial in Guruswami--Sudan decoding.
Our approach is closely related to the interpolation algorithm proposed by Lee and O'Sullivan \cite{lee_list_2009}, where a satisfactory interpolation polynomial is found as a minimal element in a certain \Grobner basis.
In \cite{beelen_efficient_2010} Beelen and Brander reformulated the interpolation problem in terms of matrices with coefficients in $\Fqq[x]$.
The advantage of this reformulation is that the interpolation problem then reduces to solving a module minimisation problem, i.e., finding a minimal weighted-degree vector in the $\Fqq[x]$-row space of a certain explicit matrix.
The \Grobner basis algorithm in this reformulation then is replaced by a weighted row reduction algorithm.
Beelen and Brander \cite{beelen_efficient_2010} improved in this way the complexity of finding the interpolation polynomial given in \cite{lee_list_2009} by applying Alekhnovich's row reduction algorithm \cite{alekhnovich_linear_2005}.
For one-point Hermitian codes they obtained a decoding algorithm with quadratic complexity in the length of the code.

Instead of using Alekhnovich's row reduction algorithm, we propose to apply the row reduction algorithm by Giorgi, Jeannerod and Villard  (GJV) \cite{giorgi_complexity_2003}.
It turns out that a straightforward application of this algorithm on the explicit matrix given in \cite{beelen_efficient_2010} does not improve complexity.
However, using a different embedding than in \cite{beelen_efficient_2010} to reformulate the interpolation problem in terms of matrices with coefficients in $\Fqq[x]$, we do find an improvement.
The result is a sub-quadratic time algorithm to find the interpolation polynomial.
By describing a fast way to deal with the so-called root-finding step (based on the theory of power series and the root-finding algorithm in \cite{alekhnovich_linear_2005}), this results in a sub-quadratic realization of the Guruswami--Sudan algorithm for one-point Hermitian codes: $\Oapp(n^{(2+\omega)/3}\ell^\omega s)$, where $s$ and $\ell$ are the multiplicity and list size parameters of Guruswami--Sudan, and $\omega \leq 3$ is the exponent for matrix multiplication.
Here and later, $\Oapp$ denotes $O$ with $\log$-factors omitted.

Next we give a new derivation of Power decoding of one-point Hermitian codes, inspired by Gao decoding for Reed--Solomon codes \cite{nielsen_power_2014}, and show how to solve the resulting generalised key equation system in a fast way.
This gives rise to a second sub-quadratic complexity decoding algorithm: $\Oapp(n^{(2+\omega)/3}\ell^\omega)$, where $\ell$ is the ``powering'' parameter.

The methodology employed here applies equally well to the classical syndrome key equation of one-point Hermitian codes used in \cite{sakata_generalized_1995} for decoding up to half the minimum distance minus half the genus.
Our results therefore puts that approach into a simple and well-studied computational framework yielding several algorithms with better complexity than in \cite{sakata_generalized_1995}.

The article is organised as follows: In \cref{sec:codes}, the necessary background is given on one-point Hermitian codes as well as on solving the Lagrange interpolation problem over the Hermitian function field.
In \cref{ssec:modules}, module minimisation is explained, which will form the core behind the fast decoding methods described later in the article.
In \cref{ssec:mod_weights}, an essential ingredient is presented, namely the embedding that will be used to reformulate the interpolation step in the decoding of one-point Hermitian codes to a module minimisation problem.

In \cref{sec:gs}, module minimisation is applied to the interpolation step in the Guruswami--Sudan list decoding algorithm for one-point Hermitian codes and a sub-quadratic algorithm is obtained in this way.
By improving existing methods to deal with the root-finding part of the Guruswami--Sudan list decoding algorithm, this leads to a complete, sub-quadratic decoding algorithm.
We first give an introduction to the Guruswami--Sudan list decoding algorithm. Subsequently, in \cref{ssec:gsQ}, the interpolation step in this algorithm is reformulated as a module minimisation problem and the techniques from \cref{ssec:modules} are applied to solve this problem in sub-quadratic time.
Then in \cref{ssec:gs_root}, the root-finding problem is discussed.

Another decoding algorithm is described in \cref{sec:power}.
``Powered key equations'' are given in \cref{ssec:keyeq}, while again the module minimisation techniques from \cref{ssec:modules} are applied to solve them in \cref{ssec:keyeq_solving}, leading to a sub-quadratic ``power decoding'' algorithm.

We have implemented the decoding algorithms in Sage v6.4 \cite{stein_sage_????} and present some simulation results in \cref{sec:simulation}: we discuss the failure probability of either decoding method, as well as the speed of the algorithm on concrete parameters.

We finish the main part of the article with some concluding remarks in \cref{sec:conclusion}.
Both in the root finding step in the Guruswami--Sudan algorithm as in an important division step in the power decoding algorithm, we need some technical machinery involving power series as well as some other technical results.
These are explained in the appendices.

\section{One-point Hermitian codes}
\label{sec:codes}

Let $q$ be some prime power, and consider the curve $\Herm$ over the field $\Fqq$ defined by the following polynomial in $X,Y$:
\[
  \Herm(X,Y) = Y^q + Y - X^{q+1}.
\]
$\Herm$ is the Hermitian curve, and it is absolutely irreducible.
Let $F = \FF {q^2}(x,y)$ be the algebraic function field with full constant field $\Fqq$ achieved by extending $\Fqq(x)$ with a variable $y$ satisfying the relation $\Herm(x,y) = 0$.
For any divisor $D$, we denote by $\L(D)$ the Riemann--Roch space associated to $D$.

There are certain basic facts about $F$ which we will need. They can be found in for example \cite{stichtenoth_note_1988}.
\begin{proposition}
  \label{prop:herm_basics}
  The function field $F$ has genus $g = \half q(q-1)$ and $q^3 + 1$ rational places, which we will denote $\mathcal P = \{ P_1,\ldots,P_{q^3},P_\infty \}$.
  The place $P_\infty$ denotes ``the place at infinity'' being the only rational place occurring as a pole of either $x$ or $y$ (in fact it is a pole of both).
  The place $P_\infty$ is totally ramified in the extension $\Fqq(x,y)/\Fqq(x)$ of function fields and hence has ramification index $q$ in this extension.
  Furthermore define
  \[
  \Ring = \bigcup_{i=0}^\infty \L(iP_\infty).
  \]
  Then $\Ring = \Fqq[x,y]$.

  Let $\PlacesAff = \Places \setminus \{ P_\infty \}$.
  By a slight abuse of notation, we can identify elements of $\PlacesAff$ with pairs $(\alpha, \beta) \in \Fqq^2$.
  For any $\alpha$, let $B_\alpha \subset \Fqq$ be the set of $\beta$ such that $(\alpha,\beta) \in \PlacesAff$.
  Then $|B_\alpha| = q$ for all $\alpha$.
  Furthermore, we have $\Div{x - \alpha} = \sum_{\beta \in B_\alpha} (\alpha,\beta) - qP_\infty$.
\end{proposition}
The fact that $\Ring = \Fqq[x,y]$ is extremely helpful since all these functions can then be described by polynomials.
For brevity, we define for any divisor $D$ the convenient notation
\[
\L(D+\infty P_\infty)=\bigcup_{i \in \ZZ} \L(D+i P_\infty).
\]
Note that for instance $\Ring = \L(\infty P_\infty)$.

For a function $f \in \Ring$ expressed as polynomials, we can therefore reduce its $y$-degree to less than $q$ using the relation $\Herm(x,y) = 0$ from which it follows that $\{ x^i y^j \mid 0 \leq j < q \}$ is a basis for $\Ring$.
We will refer to this as the ``standard basis'' of $\Ring$, and usually represent its elements using this.
However, for certain auxiliary calculations we will convert into other representations; the details of these calculations are given in \cref{app:power_series}.

We will measure elements of $\Ring$ by their pole order at $P_\infty$; when elements in $\Ring$ are in the standard basis, this takes on a particularly simple form:
\begin{definition}
  \label{def:order}
  Let the order function $\order: \Ring \mapsto \NN_0 \cup \{ -\infty \}$ be given as $\order(p) = -v_{P_\infty}(p)$ for $p \neq 0$ and $\order(0) = -\infty$, where $v_P(\cdot)$ is the valuation of a function at the place $P$.
  For a monomial $x^i y^j$, this is also given by
  \[
    \order(x^i y^j) = \deg_{q, q+1}(x^i y^j) = qi + (q+1)j,
  \]
  when $j < q$, and then extended to polynomials of $y$ degree less than $q$ by the maximal of the monomials' $\order$.
\end{definition}

Note that all monomials $x^i y^j$ with $j < q$ have different $\order$.
Therefore, $\order$ induces a term ordering $\leq_{\Herm}$ on $\Fqq[x,y]$ such that $x^{i_1}y^{j_1} \leq_{\Herm} x^{i_2}y^{i_2}$ if and only if $\order(x^{i_1}y^{j_1}) \leq \order(x^{i_2}y^{j_2})$.
This means that we can speak of the leading monomial, $\LM_\Herm(\cdot)$, and the leading coefficient, $\LC_\Herm(\cdot)$, for elements of $\Ring$.

We will also need two easy technical lemmas; the first is straightforward but a proof can be found e.g.~in \cite[Proposition 2.2]{brander_interpolation_2010}.
\begin{lemma}
  \label{lem:gs_div}
  For any non-zero $h \in F$ it holds that
  \begin{equation}
    \label{eqn:gs_div}
    \L(-\Div{h} + \infty P_\infty) = h \Ring.
  \end{equation}
\end{lemma}

\begin{lemma}
  \label{lem:gs_monoms}
  For any $m \in \ZZ_{+}$, there are at least $m-g$ distinct monomials of the form $x^i y^j$, $j < q$ such that $\order(x^i y^j) < m$.
\end{lemma}
\begin{proof}
  The statement translates simply to $\dimL((m-1)P_\infty) \geq m - g$, which is exactly Riemann's Theorem, see e.g.~\cite[Theorem 1.4.17]{stichtenoth_algebraic_2009}.
\end{proof}

Let us now formally introduce the class of codes we wish to decode.
\begin{definition}
  \label{def:gs_herm}
  Let $n=q^3$ and $m$ be an integer satisfying $2g - 2 < m < n$.
  Then the corresponding one-point Hermitian code over $\Fqq$ is defined as
  \[
    \Code = \left\{ \big( f(P_1),\ldots, f(P_n) \big) \mid f \in \L(m P_\infty) \right\}.
  \]
\end{definition}
Note that $\L(m P_\infty) \subset \Ring$, so all the $f$ we need to evalute to obtain $\Code$ are polynomials in $x$ and $y$ satisfying $\order f \leq m$.

The basic parameters of these codes are completely known. First of all from \cite[Theorem 2.2.2]{stichtenoth_algebraic_2009} it follows that in the context of \cref{def:gs_herm}, $\Code$ is an $\params n k d$ code where
  \begin{IEEEeqnarray*}{rCl+C+rCl}
    k &=& m - g + 1   & \makebox{and} &
    d &\geq& \dG \defeq n - m.
  \end{IEEEeqnarray*}
In fact, the \emph{exact} minimum distance is known: Stichtenoth showed that it is exactly $\dG$ as above whenever $2g \leq m \leq n-q^2$ \cite{stichtenoth_note_1988}, while the remaining cases were determined by Yang and Kumar and shown to be slightly better for some values of $m$ \cite{yang_true_1992}.

As a last tool before we begin, we will also need Lagrangian interpolation over the evaluation points of a considered one-point Hermitian code, i.e.~given $\gamma_{\alpha,\beta} \in \Fqq$ for every $(\alpha,\beta) \in \PlacesAff$ then find some $p \in \Ring$ such that $p(\alpha,\beta) = \gamma_{\alpha,\beta}$ for all $(\alpha,\beta)$.
It is easy to see such a function must exist: for each place, the requirement specifies a linear equation in the coefficients of $p$ seen as an element of $\Fqq[x,y]$, so by \cref{lem:gs_monoms} there must exist one with $\order$ less than $n+g+1$.
Since it is slow to solve a linear system of equations, it is beneficial to have a closed formula though this might yield a function of slightly suboptimal $\order$.
The following lemma is inspired by a similar result from \cite{lee_list_2009}, though the complexity analysis is new.
\begin{lemma}
  \label{lem:gs_interpol}
  Given $\gamma_{\alpha, \beta} \in \Fqq$ for all $(\alpha,\beta) \in \PlacesAff$ the function
  \[
    p = \sum_{\alpha \in \Fqq}
                 \prod_{\alpha' \in \Fqq \setminus \{ \alpha \}} \frac {x-\alpha'}{\alpha-\alpha'}
                 \sum_{\beta \in B_\alpha}
                    \left(\gamma_{\alpha,\beta} \prod_{\beta' \in B_\alpha \setminus \{ \beta \}}
                       \frac {y-\beta'}{\beta-\beta'}
                    \right)
  \]
  satisfies $p(\alpha,\beta) = \gamma_{\alpha,\beta}$ for $(\alpha,\beta) \in \PlacesAff$ and $\order p < n + 2g$.
  Furthermore, given the $\gamma_{\alpha,\beta}$ we can compute $p$ in time $\Oapp(n)$.
\end{lemma}
\begin{proof}
  Clearly, the given $p \in \Ring$, and first statement is easy to see.
  For the $\order$, clearly $\deg_x p \leq q^2-1$ and $\deg_y p \leq q-1$ and so $\order(p) \leq q(q^2-1) + (q+1)(q-1)$.

  For the complexity, we use standard Divide \& Conquer tricks.
  Denote by $L[B, \vec \eta](y)$ the $\Fqq[y]$ Lagrange interpolation polynomial such that $L[B, \vec \eta](\beta) = \eta_\beta$ for all $\beta \in B$.
  Note that we have $L[B, \vec \eta] = \sum_{\beta \in B} \left(\eta_{\beta} \prod_{\beta' \in B \setminus \{ \beta \}} \frac {y-\beta'}{\beta-\beta'} \right)$.
  Let $\tilde{\vec \gamma}_\alpha = \big(\gamma_{\alpha,\beta} / \prod_{\alpha' \in \Fqq \setminus \{ \alpha \}}(\alpha-\alpha') \big)_{\beta \in B_\alpha}$ for each $\alpha \in \Fqq$.
  Let $A = \Fqq$ and consider a subdivision into two disjoint sets $A_1$ and $A_2$.
  Then
  \begin{IEEEeqnarray*}{rCl}
    p &=& \sum_{\alpha \in A}
           \prod_{\alpha' \in \Fqq \setminus \{ \alpha \}} (x-\alpha')
           \ L[B_\alpha, \tilde{\vec \gamma}_\alpha](y)
    \\
    &=& \sum_{K = {1,2}} \prod_{\alpha \in A \setminus A_K} (x-\alpha) \\
    & & \quad \left (\sum_{\alpha \in A_K} \prod_{\alpha' \in A_K \setminus \{ \alpha \}} (x-\alpha') \ L[B_\alpha, \tilde{\vec \gamma}_\alpha](y) \right) ,
  \end{IEEEeqnarray*}
  Now the inner parenthesis is a recursive $\Ring$ Lagrange interpolation problem with half as many points.
  If we denote by $T(t)$ the cost of solving this problem with $qt$ points having $t$ different $x$-coordinates, we get the recursive equation for $t > 1$ that $T(t) = 2T(t/2) + q\Oapp(t/2)$: to collect the two recursive $\Ring$ functions we must perform $2q$ multiplications in $\Fqq[x]$ with operands of degree at most $t/2$, followed by $q$ sums.
  This has the solution $T(t) = \Oapp(qt) + \Oapp(t)T(1)$, where $T(1)$ then consists of computing a single $L[B_\alpha, \hat{\vec\gamma}]$ for some $\alpha$ and $\hat{\vec\gamma}$.
  This can be done in cost $\Oapp(q)$ since $|B_\alpha| = q$.
  The constants $\prod_{\alpha' \in \Fqq \setminus \{ \alpha \}}(\alpha-\alpha')$ for the $\tilde{\vec \gamma}_\alpha$ can be precomputed using Divide \& Conquer methods in time $\Oapp(q^2)$.
\end{proof}

\section{Module Minimisation}
\label{ssec:modules}

In both our algorithms, we will need to find ``small'' elements in certain free $\Fqq[x]$-modules, given a basis of the module.
We will solve this by representing the basis as a square $\Fqq[x]$ matrix and then bring it to a certain standard form; the resulting matrix will still represent a basis of our module, and its rows will represent ``small'' elements.
As a measure for being ``small'' we will use the quantity
\[
  \deg \vec v = \max_i\{ \deg v_i \},
\]
with $\vec v = (v_1,\dots,v_\rho) \in \Fqq[x]^\rho$. In this section, we will describe this process from the point where a basis $\{\vec v_1, \ldots, \vec v_\rho\}$ of an $\Fqq[x]$-module $\Mod V$ is given, in a manner completely detached from the coding theoretic setting.
We will restrict ourselves to the case that the $\vec v_i$ can be represented as $\Fqq[x]$ vectors of length $\rho$.
Let $V \in \Fqq[x]^{\dimsq \rho}$ be the matrix whose rows are the $\vec v_i$. By slight abuse of language we will sometimes also call $V$ a basis of $\Mod V$.

By ``leading position'', or $\LP(\vec v)$ for some $\vec v \in \Fqq[x]^\rho$, we mean the right-most position $i$ such that $\deg v_i = \deg \vec v$.
The problem we are going to solve is the following:
\begin{problem}
  \label{prob:minimal_vector}
  Let $I \subseteq \{ 1, \ldots, \rho \}$ and let $\Mod V_I$ be all vectors of $\Mod V$ with leading position in $I$.
  Find then a vector $\vec v \in \Mod V_I$ with minimal degree.
\end{problem}
For the Guruswami--Sudan interpolation, we will set $I = \{ 1, \ldots, \rho \}$ and will just seek any vector of minimal degree, while for Power decoding, $I$ will be only the first few indices.

\begin{definition}
  A matrix $U \in \Fqq[x]^{\dimsq \rho}$ is in \emph{weak Popov form} if the leading position of all its rows are different.
\end{definition}

Note that the weak Popov form is not canonical for a given matrix.
The following well-known result describes why the definition is so useful:
\begin{proposition}
  \label{prop:wpf}
  Let $U \in \F[x]^{\dimsq \rho}$ be a basis in weak Popov form of a module $\Mod V$.
  Any non-zero $\vec b \in \Mod V$ satisfies $\deg \vec u \leq \deg \vec b$, where $\vec u$ is the row of $U$ with $\LP(\vec u) = \LP(\vec b)$.
\end{proposition}
A proof can be found in e.g.\cite{nielsen_solving_2014}.

Using elementary row operations, we may change $V$ into a matrix $U$ without changing the row space of the matrices.
The matrices $U$ and $V$ are unimodular equivalent, that is to say that there exists $M \in \Fqq[x]^{\dimsq \rho}$ with $\det M \in \Fqq^*$ such that $U=MV$.
Clearly then, if we can compute from $V$ a unimodular equivalent matrix $U$, which is in weak Popov form, then by the above proposition we have solved our problem for any index set $I$.
This computation is known as module minimisation, $\Fqq[x]$-lattice basis reduction or row reduction\footnote{%
  These names sometime refer to computing a ``row reduced'' matrix which is a slightly weaker property than being in weak Popov form.
}.
It is well-known that the weak Popov form is also a Gr\"obner basis of the module for a specific monomial ordering, see e.g.~\cite[Section 2.1.2]{nielsen_list_2013}.

There are a number of algorithms from the literature for carrying out this computation.
Principally, their running time depends on $\deg V=\max_i\{ \deg \vec v_i \}$ where $V$ is the input matrix.
It was shown in \cite[Chapter 2]{nielsen_list_2013} how two algorithms, Mulders--Storjohann's \cite{mulders_lattice_2003} and Alekhnovich's \cite{alekhnovich_linear_2005}, rather depend on the \emph{orthogonality defect}:
\[
  \OD V = \rowdeg V - \deg\det V \leq \rho\deg V,
\]
with $\rowdeg V=\sum_{i=1}^\rho \deg \vec v_i$ and $\deg\det V$ the degree of the determinant of the matrix $V \in \Fqq[x]^{\rho \times \rho}$. \cref{tbl:wpf_complexities} summarises the complexities for module minimisation using various known algorithms.
It should be noted that the two algorithms GJV \cite{giorgi_complexity_2003} and Zhou--Labahn \cite{zhou_computing_2012} compute \emph{order bases} of $\Fqq[x]$ matrices; it was described in \cite{giorgi_complexity_2003} how to use an order basis computation to compute a row reduced form, and in \cite{sarkar_normalization_2011} how to quickly compute the weak Popov form from a row reduced one.
The asymptotic complexities are as reported for the entire sequence of algorithms.

\begin{table}[tb]
  \centering
    \begin{tabular}{l@{\hspace{3em}}l}
      \multicolumn{2}{c}{Complexity for computing a weak Popov form of $V \in \Fqq[x]^{\rho\times \rho}$}                                                   \\
      \toprule
      Algorithm                                                                    & Field operations in big-$\Oapp$ \\ \midrule
      Mulders--Storjohann \cite{mulders_lattice_2003}                              & $\rho^2\deg V\OD V$                \\
      Alekhnovich \cite{alekhnovich_linear_2005}                                   & $\rho^\omega \OD V$                \\
      GJV \cite{giorgi_complexity_2003} or Zhou--Labahn \cite{zhou_computing_2012} & $\rho^\omega \deg V$               \\
      \bottomrule
    \end{tabular}
    \caption
    {We use $\omega$ for the exponent for multiplication of $\Fqq$ matrices, i.e.~$\omega \leq 3$.
     We assume that $\rho < \deg V$.
    }
    \label{tbl:wpf_complexities}
\end{table}

\subsection{Handling Weights}
\label{ssec:mod_weights}

For application to the decoding algorithms we present later in the article, \cref{prob:minimal_vector} is not formulated quite general enough: rather, we will be seeking a vector of $\Mod V$ whose \emph{weighted} degree is minimal, and this weighting takes a rather general form: let $\nu \in \ZZ_+$ and $\vec w \in \NN_0^\rho$, then the $(\nu, \vec w)$-weighted degree of some $\vec v \in \Fqq[x]^\rho$ is
\[
  \deg_{\nu, \vec w} \vec v = \max_i\{ w_i + \nu \deg v_i \},
\]
where $v_i$ and $w_i$ are the elements of $\vec v$ respectively $\vec w$.
Similarly, we will consider $\LP_{\nu,\vec w}(\vec v) = \max \{ i \mid w_i + \nu \deg v_i = \deg_{\nu, \vec w} \vec v \}$.
For decoding one-point Hermitian codes, we will be using $\nu = q$.

We will now explain how to handle such weights without changing the underlying module minimisation algorithm or incurring any serious performance penalty.
We will introduce two injective mappings for matrices such that finding a weak Popov form of the image of $V$ under either will solve the weighted minimisation problem.
The first is a straightforward embedding of the weights but has two downsides: it can only be used with certain module minimisation algorithms, and those algorithms need to be implemented in a specific manner to avoid a computational overhead.
To mitigate both of these problems we derive a second embedding from the first.

First we define the following straightforward map $\Wmap: \Fqq[x]^\rho \mapsto \Fqq[x]^\rho$:
\begin{IEEEeqnarray*}{rCl}
  \Wmap\big( (v_1,\ldots, v_\rho)\big) &=& \big( x^{w_1}v_1(x^\nu), \ldots, x^{w_\rho}v_\rho(x^\nu) \big).
\end{IEEEeqnarray*}
We extend $\Wmap$ row-wise to $\dimsq{\rho}$ matrices such that the $i$th row of $\Wmap(V)$ is $\Wmap(\vec v_i)$, where $\vec v_i$ are the rows of $V$.
Note that $\Wmap(\Mod V)$ is a free $\Fqq[x^{\nu}]$-module of dimension $\rho$, and that any basis of it is by $\Wmap^\mo$ sent back to a basis of $\Mod V$.

\begin{proposition}
  A vector $\vec v \in \Mod V$ has minimal $\deg_{\nu, \vec w}$ if $\Wmap(\vec v)$ has minimal degree in $\Wmap(\Mod V)$.
  Furthermore, $\LP_{\nu, \vec w}(\vec v) = \LP(\Wmap(\vec v))$.
\end{proposition}
\begin{proof}
  This follows immediately since for any vector $\vec v = (v_1,\ldots,v_\rho) \in \F[x]^\rho$, then $\deg \Wmap(\vec v) = \deg_{\nu, \vec w} \vec v$.
\end{proof}

In other words, we can hope to solve the weighted problem as follows: find a $\Wmap(W)$ in weak Popov form and unimodular equivalent to $\Wmap(V)$.
Then the $\Wmap^\mo$-map of the row of $\Wmap(W)$ with minimal degree and leading position in $I$ yields the sought solution.
However, a general module minimisation algorithm will consider the $\Fqq[x]$-module spanned by $\Wmap(V)$ -- and not the $\Fqq[x^\nu]$-module -- so a weak Popov form of $\Wmap(V)$ will generally not result in a matrix in $\Wmap(\Mod V)$, and hence we cannot apply $\Wmap^\mo$ to its rows.
In the case of the Mulders--Storjohann algorithm \cite{mulders_lattice_2003} or the Alekhnovich algorithm \cite{alekhnovich_linear_2005}, one can show that things will go well: applying either algorithm to $\Wmap(V)$ results in a weak Popov form in $\Wmap(\Mod V)$ \cite{lee_list_2009,beelen_efficient_2010,nielsen_generalised_2013}.
Furthermore, if properly implemented, these algorithms will not incur a computational penalty from the $x \mapsto x^\nu$ blow-up.

To take advantage of the faster module minimisation algorithms -- in a manner ensuring both correctness and speed -- we introduce a second mapping which does not have the problems of $\Wmap$.

For this improved mapping, consider first the permutation $\pi$ of $[1,\ldots,\rho]$ defined indirectly by the following property:
\begin{IEEEeqnarray*}{rCl}
  \pi(i) &>& \pi(j) \iff \ (w_i \modop \nu) > (w_j \modop \nu) \\
   && \qquad \qquad \lor \,\big( (w_i \modop \nu) = (w_j \modop \nu) \land i > j \big).
\end{IEEEeqnarray*}
The permutation $\pi$ acts on vectors of $\F[x]^\rho$ by permuting the positions of such vectors.
Our desired mapping is now $\Pmap$:
\[
  \Pmap\big( (v_1,\ldots,v_\rho) \big) = \pi\big((x^{\floor{w_1/\nu}} v_{1}, \ldots, x^{\floor{w_\rho/\nu}} v_\rho) \big).
\]
\begin{proposition}
  \label{prop:weight_mappings}
  For any $\vec v \in \F[x]^\rho$ then
  \[
    (\pi^\mo \circ \LP \circ \Pmap)(\vec v) = (\LP \circ \Wmap)(\vec v).
  \]
\end{proposition}
\begin{proof}
  Let $v_i$ be the elements of $\vec v$, and $h = (\LP \circ \Wmap)(\vec v)$.
  We will prove that no index but $\pi(h)$ can be $(\LP \circ \Pmap)(\vec v)$.
  Consider first some $j > h$.
  By the definition of $h$ then
  \begin{IEEEeqnarray}{rCl}
    \nu \deg v_h + w_h &>& \nu \deg v_j + w_j, \quad \word{i.e.}            \label{eqn:map_perm_as_nu_ineq}\\
    \deg v_h + \floor{w_h/\nu} + \tfrac{w_h \modop \nu} \nu &>& \deg v_j + \floor{w_j/\nu} + \tfrac{w_j \modop \nu} \nu.
    \notag
  \end{IEEEeqnarray}
  So either $\deg v_h + \floor{w_h/\nu} > \deg v_j + \floor{w_j/\nu}$, or they are equal and $w_h \modop \nu > w_j \modop \nu$.
  In the first case, then clearly $\pi(j)$ cannot be $(\LP \circ \Pmap)(\vec v)$ due to degrees.
  In the second case the degrees of $\Pmap(\vec v)$ at positions $\pi(h)$ and $\pi(j)$ are tied, but we have $\pi(h) > \pi(j)$, which means that $\pi(j)$ cannot be the leading position.

  Consider now some $j < h$, so we have the same inequality \eqref{eqn:map_perm_as_nu_ineq} but with $>$ replaced by $\geq$.
  If sharp inequality really holds, then we can continue as before, so assume instead that equality holds.
  That implies both $\deg v_h + \floor{w_h/\nu} = \deg v_j + \floor{w_j/\nu}$ and $w_h \equiv w_j \mod \nu$.
  So the degrees of $\Pmap(\vec v)$ at positions $\pi(h)$ and $\pi(j)$ are tied, but then since $h > j$, we have $\pi(h) > \pi(j)$.
  Again, $\pi(j)$ is not the leading position.
\end{proof}
\begin{corollary}
  For any $V \in \F[x]^\rho$, then $\Wmap(V)$ is in weak Popov form if and only if $\Pmap(V)$ is in weak Popov form.
\end{corollary}

The algorithm is then clear: to solve the weighted minimisation problem, simply compute a weak Popov form of $\Pmap(V)$.
The row with minimal degree, and in case of a tie least $\LP$, only among rows whose $\LP$ are in $\{ \pi(i) \mid i \in I \}$ corresponds to a minimal solution, and one applies $\Pmap^\mo$ to obtain the vector of $\Mod V$.
This works immediately for any module minimisation algorithm.

For calculating the resulting complexity in general, one observes that
\[
\deg \Pmap(V) \leq \gamma \defeq \deg V + \max_{j}(w_j/\nu).
\]
The trivial bound gives $\OD{\Pmap(V)} \leq \rho\gamma$, so in \cref{tbl:wpf_complexities}, one can replace $\deg V$ with $\gamma$ and $\OD{V}$ with $\rho\gamma$ to obtain the generic complexities for solving the weighted problem.

\section{Fast Implementation of Guruswami--Sudan}
\label{sec:gs}

We will now present a sub-quadratic realisation of the Guruswami--Sudan decoding algorithm for the one-point Hermitian codes introduced in \cref{sec:codes}.
The main contribution is demonstrating how to perform the interpolation step using the fast module minimisation techniques discussed in the previous section.
This builds heavily on previous works \cite{lee_list_2009,beelen_efficient_2010}, and we remark further on this at the end of \cref{ssec:gsQ}.
Since the fastest previously known method for performing the root finding step was at least quadratic in $n$ \cite{beelen_efficient_2010}, we also describe how to sufficiently speed up this step in \cref{ssec:gs_root}.

In the following sections, we will consider dealing with a particular choice of a one-point Hermitian code, and use all the introduced variables $n, k, P_i, \dG, \Code$, etc.~from \cref{sec:codes}.
We will consider that a given codeword $\vec c \in \Code$ was sent, resulting from evaluating $f \in \L(mP_\infty)$, and that $\vec r = \vec c + \vec e$ was received with some error $\vec e$.
Further denote by $\Errs$ the set of error positions, i.e.~$\Errs = \{ i \mid e_i \neq 0 \}$.
The aim is to recover $\vec c$ knowing only $\vec r$, possibly even when $\weight{\vec e} \geq \dG\!/2$.

We will be working with elements of $\Ring[z]$, i.e.~the univariate polynomial ring over $\Ring$.
Define for such $Q = \sum_{t=0}^{\deg_z Q} Q_t(x,y) z^t \in \Ring[z]$ the coefficient-selecting notation $\Csel Q t$ to mean $\Csel Q t = Q_t(x,y) \in \Ring$.
We extend our degree function in a natural way to $\orderz w$ for any $w \in \RR$, so that some $Q \in \Ring[z]$ has $\orderz w Q = \max_t\{ \order \Csel Q t + tw \}$.
\begin{definition}
  \label{def:gs_zero_multi}
  A polynomial $Q \in \Ring[z]$ has a zero $(P,z_0) \in \Places^\star \times \Fqq$ with multiplicity at least $s$ if $Q$ can be written as $\sum_{j+h \geq s} \gamma_{j,h} \phi^j(z-z_0)^h$ for some $\gamma_{j,h} \in \Fqq$, where $\phi$ is a local parameter for $P$.
\end{definition}

For any place $(\alpha,\beta)$, one can choose as local parameter $\phi=x-\alpha$, which makes the above definition easy to operate with.
Note though that the sum in $Q=\sum_{j+h \geq s} \gamma_{j,h} \phi^j(z-z_0)^h$ will be an infinite sum (that is to say, a power series) in general.
However, to determine whether or not the multiplicity of $((\alpha,\beta),z_0)$ is at least $s$, one only needs to compute finitely many terms of this power series.
For one-point Hermitian codes, the Guruswami--Sudan algorithm then builds on the following theorem:

\begin{theorem}[Guruswami--Sudan]
  \label{thm:gs_herm}
  Let $s, \ell, \tau \in \ZZ_+$ be given.
  If a non-zero $Q \in \Ring[z]$ with $\deg_z Q \leq \ell$ satisfies
  \begin{enumerate}
    \item $Q$ has a zero at $(P_i, r_i)$ with multiplicity at least $s$ for $i = 1,\ldots,n$,
            \label{itm:gs_Qreq_zeroes}
    \item $\orderz m Q < s(n-\tau)$
            \label{itm:gs_Qreq_degree}
  \end{enumerate}
  and if $|\Errs| \leq \tau$, then $Q(f) = 0$.
\end{theorem}

Note that $\ell$ is to be given as an a priori bound on $\deg_z Q$, but another bound is already indirectly enforced by \cref{itm:gs_Qreq_degree}: by this, it never makes sense to choose $\ell$ such that $s(n-\tau) - \ell m \leq 0$.

\begin{remark}
  An analogous theorem holds for much more general AG codes, though $\Ring$ of course needs to be defined properly.
  See e.g.~\cite{guruswami_improved_1999} or the expository description in \cite{beelen_decoding_2008}.
\end{remark}

One can find a satisfactory $Q$ by solving a system of linear equations in the $\Fqq$-coefficients for its $x^iy^jz^h$-monomials, and one can ensure that this system will have a non-zero solution by satisfying a certain expression in the parameters.
The resulting equation can be analysed for determining the maximal $\tau$ and corresponding choices of $s$ and $\ell$.
We are not going to perform that analysis but see e.g.~\cite{lee_list_2009}.
Given $s$ and $\ell$, one can use the equation to compute a value $\tauGS(s,\ell)$ such that one can choose any $\tau \leq \tauGS(s,\ell)$.
Furthermore, $\tauGS(s,\ell)$ is the greatest integer less than
\begin{equation}
  \label{eqn:gs_decoding_radius}
  \left(1 - \frac{s+1}{2(\ell+1)}\right)n - \frac m 2 \frac \ell s  - \frac g s \ .
\end{equation}
Analysing the asymptotics of this bound, one sees that there are choices of $s$ and $\ell$ which allows choosing any $\tau < n - \sqrt{n(n-\dG)}$.
This function $n - \sqrt{n(n-\dG)}$ is called the Johnson radius.

For specific parameters of the code and $s$ and $\ell$, the lower bound on $\tauGS(s,\ell)$ is good but not always tight; it is easy to compute the precise value of $\tauGS(s,\ell)$, though a closed expression is complicated.
If one considers the Guruswami--Sudan as an algorithm taking $s$ and $\ell$ as parameters (and the code), then 
$\tauGS(s,\ell)$ is the guaranteed number of errors that it is able to correct.
It is very interesting that the algorithm will quite often succeed in correcting more errors; this was already remarked in \cite{lee_list_2009}.
We will get back to this in \cref{sec:simulation}.

\subsection{Finding $Q$ in an Explicit Module}
\label{ssec:gsQ}

We will now concern ourselves with the problem of finding $Q$.
We will assume $s \leq \ell$; with the proper analysis of choices of $s$ and $\ell$, one can show that $s > \ell$ implies $\tau < \dG\!/2$.

\begin{definition}
  \label{def:gs_mod}
  Let $\Mod M_\sell \subset \Ring[z]$ denote the set of all $Q \in \Ring[z]$ such that $Q$ has a zero of multiplicity $s$ at $(P_i, r_i)$ for $i=1,\ldots,n$, and $\deg_z Q \leq \ell$.
\end{definition}
Finding a $Q \in \Ring[z]$ for satisfying the requirements of \cref{thm:gs_herm} is then the same as finding an element in $\Mod M_\sell$ with low enough $\orderz m$.
We will find one with minimal $\orderz m$ which is guaranteed to be sufficient by the choice of parameters $s,\ell,\tau$.

It is not hard to see that $\Mod M_\sell$ is a $\Ring$-module.
To proceed, we will need to give an explicit basis for $\Mod M_\sell$. We will use a basis previously given in the literature \cite{lee_list_2009}. We will need two functions in $\Ring$:
\begin{IEEEeqnarray}{rCl+l/rCl?l}
  G &=& \prod_{i=1}^{n/q}(x - \alpha_i)   = x^{q^2}-x, &\\
  R: & & R(P_i) = r_i & \forall i=1,\ldots,n.
  \label{eqn:RandG}
\end{IEEEeqnarray}
The function $G$ is known in advance and by \cref{prop:herm_basics}, we have $\Div G = \sum_{i=1}^n P_i - nP_\infty$.

The function $R$ depends on the received word $\vec r$. Any non-zero function in $\Ring$ satisfying the interpolation constraints will do; we can either solve the linear system of equations in its coefficients, or  we can use the explicit formula of \cref{lem:gs_interpol}.
The desired explicit basis of $\Mod M_\sell$ is the following:
\begin{theorem}[\!\!\protect{\cite[Proposition 7]{lee_list_2009}}]
  \label{thm:gs_Mbasis}
  $\Mod M_\sell$ is generated as a $\Ring$-module by the $\ell+1$ polynomials $H\T i \in \Ring[z]$ given by
  \begin{IEEEeqnarray*}{rCl+l}
    H\T t(z) &=& G^{s-t}(z-R)^t,         & \textrm{for } 0 \leq t \leq s,\\
    H\T t(z) &=& z^{t-s}(z-R)^s,         & \textrm{for } s < t \leq \ell.
  \end{IEEEeqnarray*}
\end{theorem}
We need to project this module, its basis and the weighted degree into $\Fqq[x]$ in some sensible manner to be able to use the tools of \cref{ssec:modules} to find an element in $\Mod M_\sell$ of minimal $\order$.

Firstly, introduce $\vectify \null\!: \Ring \mapsto \Fqq[x]^q$: for any $g = \sum_{i=0}^{q-1} y^i g_i(x) \in \Ring$, then we define $\vectify g = (g_0, \ldots, g_{q-1})$.
As we have previously noted, any element of $\Ring$ can uniquely be written such that the $y$-degree is at most $q-1$. This implies that the map $\vectify \null\!$ is well-defined and a bijection.
Let $\Ring[z]_\ell$ be the set of polynomials of $z$-degree at most $\ell$; then we also introduce $\vectifyz \null\!: \Ring[z]_{\ell} \mapsto \Fqq[x]^{(\ell+1)q}$, as for any $Q \in \Ring[z]_\ell$, then $\vectifyz Q = \big( \vectify {\Csel Q 0} \mid \ldots \mid \vectify {\Csel Q \ell} \big)$.
Define now $\vec w \in \NN_0^{(\ell+1)q}$ as $\vec w = (\vec w_0 \mid \ldots \mid \vec w_{\ell})$, where
\[
  \vec w_t = (tm, tm + q+1,\ldots, tm + (q-1)(q+1)).
\]
One can then verify the following identity for any $Q \in \Ring[z]_\ell$:
\[
  \orderz m(Q) = (\deg \circ\:\WmapH \circ \vectifyz\null)(Q),
\]
where $\WmapH$ is as in \cref{ssec:modules}.

\begin{proposition}
  \label{prop:gs_Abasis}
  Let $A_{s,\ell} \in \Fqq[x]^{\dimsq{(q(\ell+1))}}$ be given as
  \[
    \left(
      \begin{array}{c|c|c}
    \left(
      \begin{array}{c}
        \vectifyz{H\T 0} \\\hline
        \vectifyz{yH\T 0} \\\hline
        \vdots \\\hline
        \vectifyz{y^{q-1}H\T 0}
      \end{array}
      \right)^{\!\!\!\top}
      & \cdots &
    \left( \begin{array}{c}
        \vectifyz{H\T \ell} \\\hline
        \vectifyz{yH\T \ell} \\\hline
        \vdots \\\hline
        \vectifyz{y^{q-1}H\T \ell}
      \end{array}
    \right)^{\!\!\!\top}
  \end{array}
\right )\Transp,
  \]
then $\Mod M_{s,\ell}$ is in bijection with the $\Fqq[x]$ row space of $A_{s,\ell}$ through the map $\vectifyz \null$.
  Let $\vectifyz Q$ be the vector in this row space with minimal $\WmapH$-weighted degree.
  Then $Q$ has minimal $\orderz m$ in $\Mod M_{s,\ell}$.
\end{proposition}
\begin{proof}
  Consider some $Q(z) \in \Mod M_{s,\ell}$; by \cref{thm:gs_herm} we can find $p_t \in \Ring$ such that $Q(z) = \sum_{t=0}^\ell p_t H\T t(z)$.
  Let $p_t = \sum_{j=0}^{q-1} p_{t,j}y^j$ with $p_{t,j} \in \Fqq[x]$, then
  \[
    Q = \sum_{t=0}^\ell p_t H\T t
      = \sum_{t=0}^\ell\sum_{j=0}^{q-1} p_{t,j}(y^jH\T t).
  \]
  This directly implies that
  \[
    \vectifyz Q = \sum_{t=0}^\ell\sum_{j=0}^{q-1} p_{t,j}\vectifyz {y^jH\T t},
  \]
  which is to say, $\vectifyz Q$ is in the $\Fqq[x]$ row space of $A_{s,\ell}$.

  The claim on weighted degrees follow immediately from $\orderz m = \deg \circ\:\WmapH \circ \vectifyz\null$.
\end{proof}

By \cref{prob:minimal_vector}, we can therefore find a minimal $\orderz m$-weighted $Q \in \Mod M_{s,\ell}$ by bringing  $\PmapH(A_{s,\ell})$ to weak Popov form.
We get:
\begin{proposition}
  In the context of \cref{prop:gs_Abasis}, the worst-case complexity of finding a satisfactory $Q$ as a minimal element in the row space of $\PmapH(A_{s,\ell})$ is as in \cref{tbl:gs_compl}, for various choices of module minimisation algorithm.
\end{proposition}
\begin{proof}
  We firstly need to construct $A_{s,\ell}$: we assume $G^t$ precomputed for $t=1,\ldots,s$, and $R$ can be computed in $\Oapp(n)$ according to \cref{lem:gs_interpol}.
  Computing $R^t$ for $t=1,\ldots,s$, each represented in the standard basis with $y$-degree less than $q$, can be done iteratively in $s\Oapp(q)\Oapp(sq^2) = \Oapp(s^2n)$: $R \cdot R^{t-1}$ can be computed as multiplying two degree $q-1$ polynomials in $y$ whose coefficients are in $\Fqq[x]$ with degree in $O(sq^2)$ by \cref{lem:gs_interpol}.
  We then need to use $\Herm$ to reduce the $y$-degree to less than $q$, which can be done with at most $3q$ additions of $\Fqq[x]$-polynomials of degree at most $O(sq^2)$.
  To then compute the $H\T t$, we need $\tbinom t {t'} G^{s-t}R^{t'}$ for $t=0,\ldots,s$ and $0 \leq t' \leq t$ at a cost of a further $\Oapp(s^3n)$.
  Since the $H\T t$ are then computed in the standard basis, the final construction of $A_{s,\ell}$ is simply linear in its size which is $O(\ell^2sn)$.
  
  By \cref{ssec:mod_weights}, the complexity of bringing $\PmapH(A_{s,\ell})$ to weak Popov form is dominated by
  \begin{IEEEeqnarray*}{rCl}
    \gamma &=& \deg(A_{s,\ell}) + \max{\vec w}/q \\
           &=& O(sn^{2/3}) + (\ell m + (q-1)(q+1))/q.
  \end{IEEEeqnarray*}
  By the note right after \cref{thm:gs_herm} then $\ell m < s(n-\tau)$ so $\gamma \in O(s n^{2/3})$.
  The complexities then follow by noting that $A_{s,\ell}$ has $(\ell+1)q$ rows and columns, and $\OD{\PmapH(A_{s,\ell})} \leq (\ell+1)q \gamma $.
\end{proof}

\begin{table}
  \centering
  \hspace*{-0.2cm}
  \def\Dist{\hspace*{0.8em}}
  \begin{tabular}{@{\Dist}l@{\Dist}l}
    \multicolumn{2}{c}{Complexity for computing $Q$}                                      \\
    \toprule
      Algorithm                                     & Field operations in big-$\Oapp$     \\ \midrule
    Mulders--Storjohann \cite{mulders_lattice_2003}                               & $n^{7/3}\ell^3s^2$                  \\
    Alekhnovich \cite{alekhnovich_linear_2005}                                    & $n^{(3+\omega)/3}\ell^{\omega+1} s$ \\
    GJV \cite{giorgi_complexity_2003} or Zhou--Labahn \cite{zhou_computing_2012}  & $n^{(2+\omega)/3}\ell^\omega s$     \\
    \bottomrule
  \end{tabular}
  \caption
  {Use of $\Oapp$ and $\omega$ as in \cref{tbl:wpf_complexities}.}
  \label{tbl:gs_compl}
\end{table}

\begin{remark}
  For the interpolation step of Guruswami--Sudan in decoding of algebraic geometry codes, both the Mulders--Storjohann and the Alekhnovich algorithm have been suggested, \cite{lee_list_2009} respectively \cite{beelen_efficient_2010}.
  Note that the algorithm described in \cite{lee_list_2009} is computationally equivalent with the Mulders--Storjohann algorithm though derived in terms of \Grobner bases.
  In both \cite{lee_list_2009} and \cite{beelen_efficient_2010}, the mapping $\Wmap$ was (implicitly) used together with a detailed analysis of the module minimisation algorithms to prove that the operations did not leave the $\F[x^q]$-module, and that the slow-down discussed in \cref{ssec:mod_weights} did not occur.

  The GJV and the Zhou--Labahn methods have not previously been applied for this decoding setting, and the application of $\PmapH$ allows us to deduce correctness and the low complexity without investigating the algorithm in detail.

  Note that the GJV has previously been suggested for decoding of Reed--Solomon codes \cite{cohn_ideal_2010}.
\end{remark}

\subsection{Fast Root finding}
\label{ssec:gs_root}

\def\local{\phi}
\def\Pring{\PowerSeriesRing[\local] \Fqq}

After having constructed $Q(z)$, we should find all $f \in \L(mP_\infty)$ such that $Q(f) = 0$.
This can be done using Hensel lifting \cite{wu_efficient_2001,beelen_decoding_2008}, inspired by the algorithm of Roth and Ruckenstein \cite{roth_efficient_2000} for solving the root-finding problem for Reed--Solomon codes.
The complexity of these methods all have at least quadratic dependence on $n$, and so would be asymptotically slower than the interpolation described in the previous section.

Alekhnovich described in \cite{alekhnovich_linear_2005} how to use fast arithmetic to bring the method of \cite{roth_efficient_2000} down to quasi-linear complexity in $n$.
Using the power series idea of \cite{beelen_decoding_2008} it is easy to apply this algorithm to our root-finding problem as well.
For our case, the main result can be paraphrased as follows; its proof as well as the complete root-finding algorithm is given in \cref{app:root_finding}.

\begin{proposition}
  For $Q \in \Ring[z]$ satisfying the requirements of \cref{thm:gs_herm}, then we can compute all $f \in \L(mP_\infty)$ such that $Q(f) = 0$ in time $\Oapp(n^{4/3}\ell^2 s)$.
\end{proposition}

We have now described how to realise the complete Guruswami--Sudan algorithm with asymptotic complexity $\Oapp(n^{(\omega+2)/3}\ell^\omega s)$.
Note that the only step in the entire algorithm with this complexity is the module minimisation step; all other steps have lower order.
This means that the hidden constant in the big-$O$ notation for the leading term in our decoder must be \emph{exactly} that of the module minimisation employed.
In an implementation and for concrete parameters, one could of course still be concerned that the remaining, asymptotically lower-order terms, dominate the actual running time.
We demonstrate in \cref{sec:simulation} that this is unlikely since their running time in our implementation is very low.

\section{Fast Power Decoding}
\label{sec:power}

In this section we will present a decoding algorithm generalising classical syndrome decoding \cite{justesen_fast_1992} for low-rate one-point Hermitian codes, obtained by ``powering'' the key equations.
The technique, also known as ``virtual extension to an interleaved code'' was developed for Reed--Solomon codes by Schmidt et al.~\cite{schmidt_syndrome_2010}.
It has already been suggested for one-point Hermitian codes by Kampf and Li \cite{kampf_bounds_2012, kampf_decoding_2013}, but no proof of the algorithm's complexity was given.

As opposed to this previous work, we will power a Gao-style key equation in place of the classical syndrome key equation.
Apart from the joy of variety, this admits a succinct derivation which follows the definition of the codes as evaluations closely, and it highlights some similarities with Guruswami--Sudan decoding.
Another advantage is that the sent information polynomial is evident immediately, and one does not need to find the zeroes of the error locator and do erasure decoding or similar afterwards.
For Reed--Solomon codes, this variation was suggested in \cite{nielsen_power_2014} and proved to be behaviourally equivalent to the syndrome formulation.

We will show how to put the problem into a framework where fast algorithms for module minimisation can be directly applied, and this will yield a fast decoding algorithm with speed asymptotically comparable to that of Guruswami--Sudan.
As with Guruswami--Sudan, one can set the decoding algorithm's parameters to perform minimum distance decoding, and in this case we improve upon the fastest, previously known techniques.
Note that the module minimisation framework also applies to classical syndrome decoding, and is therefore the first significant speed improvement of this technique in the last 20 years, since \cite{sakata_generalized_1995}.

Power decoding is not list decoding: it either gives one answer or it will fail.
For Reed--Solomon codes, it might only fail when the number of errors has exceeded half the minimum distance, and statistically this has been verified to occur only very rarely.
There are failure probability bounds for ``powering degree'' 2 and 3, but not in the general case \cite{schmidt_syndrome_2010,zeh_unambiguous_2012,nielsen_power_2014}.
For one-point Hermitian codes, the genus of the curve play a role in the decoding radius---as usual---and we will get back to the precise decoding performance in \cref{ssec:power_radius}.
As for Reed--Solomon codes, we have not yet obtained a bound on the failure probability, but experiments indicate similar behaviour.

\subsection{Key Equations}
\label{ssec:keyeq}

Recall that $\vec r = \vec c + \vec e$ was received, and denote the set of error positions by $\Errs$.

\begin{definition}
  \label{def:power_errloc}
  The \emph{error locator} $\Lambda$ is the non-zero polynomial in $\L(-\sum_{i \in \Errs} P_i + \infty P_\infty)$ with minimal $\order$ and $\LC_\Herm(\Lambda) = 1$.
\end{definition}
Clearly, $\Lambda \in \Ring$ since the defining Riemann--Roch space is a subset of $\Ring$.
It is easy to see that the definition is well-posed, i.e.~there is exactly one element in the Riemann--Roch space satisfying the restrictions.
\begin{lemma}
  \label{lem:power_errloc}
  $|\Errs| \leq \order \Lambda \leq |\Errs| + g$.
\end{lemma}
\begin{proof}
  Being in $\L(-\sum_{i \in \Errs} P_i + \infty P_\infty)$ specifies $|\Errs|$ homogeneous equations in the coefficients of $\Lambda$, so by \cref{lem:gs_monoms}, we will still have more coefficients than equations after requiring $\order \Lambda < |\Errs| + g + 1$.
  For the lower bound, then since $\deg (-\sum_{i \in \Errs} P_i + t P_\infty) < 0$ for $t < |\Errs|$ we must have $\L(-\sum_{i \in \Errs} P_i + t P_\infty) = \{0\}$ whenever $t < |\Errs|$.
  Since $\Lambda \neq 0$ is in this Riemann--Roch space when $t = \order \Lambda$, then clearly $\order \Lambda \geq |\Errs|$.
\end{proof}

Recall now $G$ and $R$ from \cref{eqn:RandG}, and extend the latter to ``powers'':
\begin{IEEEeqnarray}{l/rCl?l}
  \label{eqn:power_R}
  R\T t: & R\T t(P_i) &=& r_i^t & \forall i=1,\ldots,n,\quad t \in \NN_0.
\end{IEEEeqnarray}
Again, the $R\T t$ can be found by solving the emerging linear systems of equations or using the explicit formula of \cref{lem:gs_interpol}.
We then immediately arrive at the powered key equations over the function field:
\begin{theorem}
  \label{thm:power_power}
  $\Lambda R\T t \equiv \Lambda f^t \mod G$ for $t \in \NN_0$ as a congruence over $\Ring$.
\end{theorem}
\begin{proof}
  We have $\Lambda R\T t - \Lambda f^t = \Lambda(R\T t - f^t) \in \L(-\sum_{i=1}^n P_i + \infty P_\infty)$, since for $i \in \Errs$ then $\Lambda(P_i) = 0$ while for $i \notin \Errs$ then $R\T t(P_i) = f^t(P_i)$.
  Recall that $\Div G = \sum_{i=1}^n P_i - nP_\infty$; therefore by \cref{lem:gs_div} we must have $G \mid \Lambda(R\T t - f^t)$ over $\Ring$.
\end{proof}
This means that the sought $\Lambda$ is a solution to a list of key equations -- but over $\Ring$.
We will handle these non-linear equations similarly to how classical key equations are handled:
regard the right-hand side as unknowns independent of $\Lambda$ and each other, and only enforce bounds on its degree.
Then seek a minimal $\order$-element $\hat \Lambda \in \Ring$ such that $\hat\Lambda R\T t \modop G$ satisfies this degree bound for each $t$.
One then hopes that $\hat \Lambda = \Lambda$.

\cref{thm:power_power} provides us with infinitely many key equations, but when we are using the above strategy we are only aided by those for which  the degree bound on $\hat\Lambda R\T t \modop G$ is not trivially satisfied; in particular, when
$\order(\Lambda f^t) > \order G$ then the key equation for this $t$ is useless.
We do not know $\order \Lambda$ but we can at least disregard those equations for which $tm \geq \order G = n$.
Thus, in the following, we will use equations $t = 1,\ldots, \ell$ where $\ell$ is chosen such that $\ell m < n$.

As in \cref{sec:gs}, we will project the key equations over $\Ring$ into equations over $\Fqq[x]$ to be able to use module minimisation for finding the minimal $\hat \Lambda$.
We will introduce a bit more notation for this:
for two $a, b \in \Ring$, with vector forms $\vectify a, \vectify b$ we wish to represent their product $ab$ in vector form.
With $\vectify b = (b_0,\ldots, b_{q-1})$, consider the following vector--matrix product:
\begin{equation}
  \label{eqn:gs_vec_mult}
  \vectify a
  \left(\begin{array}{ccccccc}
     b_0                              & b_1 & \ldots & b_{q-1} &         &        & \multirow{2}{*}{\mbox{\Large 0}} \\
                                      & b_0 &  b_1   & \ldots  & b_{q-1} &                                           \\
    \multirow{2}{*}{\mbox{\Large 0}}  &     & \ddots &         &         & \ddots &                                  \\
     &     &        & b_0     & b_1     & \ldots & b_{q-1}
  \end{array}\right).
\end{equation}
The result will be a vector $(c_0, \ldots, c_{2q-2}) \in \Fqq[x]^{2q-1}$ such that $ab = \sum_{i=0}^{2q-2} c_i(x) y^i$.
Denote by $\matmult b$ the matrix of the above form, for any $b \in \Ring$.
Using the curve equation $\Herm$ to rewrite $\sum_{i=0}^{2q-2} c_i(x) y^i$ into having $y$-degree less than $q$ becomes the result of the linear transformation
\begin{equation}
  \label{eqn:herm_vec_red}
  \mtrx{ c_0 & c_1 & \ldots & c_{2q-2}}
  \left(\begin{array}{c}
      \raisebox{0pt}[5pt][5pt]{$I_{q \times q}$} \\\hline
      \begin{array}{@{}ccccc@{}}
         x^{q+1} & -1     &         &  \multirow{2}{*}{\mbox{\Large 0}} \\
                 & \ddots & \ddots  &    \\
         \multirow{2}{*}{\mbox{\Large 0}}       &        & x^{q+1} & -1 \\
                &        &         & x^{q+1}
      \end{array}
    \end{array}\right),
\end{equation}
where $I_{q \times q}$ is the $q \times q$ identity matrix.
Denote the matrix in the above product by $\Xi$.
With this notation then we can write
\begin{equation}
  \label{herm_vec_multred}
  \vectify{ab} = \vectify a \matmult b \Xi.
\end{equation}
\begin{corollary}
  \label{cor:power_keyeq_fx}
  $\vectify \Lambda = (\Lambda_0,\ldots,\Lambda_{q-1})$ satisfies the $q\ell$ congruences over $\Fqq[x]$:
  \[
    \sum_{i=0}^{q-1} \Lambda_i T_{i,j} \equiv B_j \mod G  , \quad j = 1,\ldots q\ell,
  \]
  where the $B_j \in \Fqq[x]$ satisfy
  \[
    q\deg B_j + (q+1)((j-1) \modop q) < \order \Lambda + m\ceil{j/q} + 1
  \]
  and where $T = [ T_{i,j} ] \in \Fqq[x]^{q \times q\ell}$ equals the matrix
  \[
    \big[ \matmult{R\T 1}\Xi \mid \matmult{R\T 2}\Xi \mid \ldots \mid  \matmult{R\T \ell}\Xi \big]
  \]
  element-wise reduced modulo $G$.
\end{corollary}
\begin{proof}
  \cref{thm:power_power} implies for each $t$ that there is a $p_t \in \Ring$ such that $\Lambda R\T t = \Lambda f^t + p_t G$, which means
  \begin{IEEEeqnarray*}{rCl+l}
    \vectify{\Lambda R\T t} &=& \vectify{\Lambda f^t} + G\vectify{p_t}.
  \end{IEEEeqnarray*}
  Letting $\vectify{\Lambda f^t} = (B_{t,0}, \ldots, B_{t,q-1})$, then the above implies for $h=0,\ldots,q-1$ that
  \[
    \vectify \Lambda \matmult{R\T t} \Xi \equiv B_{t,h} \mod G
  \]
  as an $\Fqq[x]$-congruence.
  Furthermore, since $\order(\Lambda f^t) \leq \order \Lambda + tm$, then $q\deg B_{t, h} + h(q+1) \leq \order \Lambda + tm$.

  Taken over all $t = 1,\ldots,\ell$ and relabelling $B_{t,h}$ appropriately, then this gives the $q\ell$ congruence equations of the corollary.
\end{proof}

Note that the degree constraints on the remainders $B_j$ depend on $\order \Lambda$, i.e.~on $\max_i\{q\deg \Lambda_i + i(q+1) \}$.
The above $q\ell$ equations therefore constitute a heavily generalised form of a weighted key equation.
The form of the ``key equation'' is elsewhere often called Pad\'e approximation, and the equations of \cref{cor:power_keyeq_fx} generalise both the notion of Simultaneous Pad\'e and Hermitian Pad\'e.
This form was recently considered in \cite{nielsen_solving_2014} under the name ``asymmetric 2D Pad\'e approximation''; see also this paper for discussion on and references to other Pad\'e-like approximants.

\subsection{Solving the Key Equations}
\label{ssec:keyeq_solving}

We will here outline the method of \cite{nielsen_solving_2014} for finding a minimal solution to the equations of \cref{cor:power_keyeq_fx}.
By ``solution'' we mean any $(\hat \Lambda_0,\ldots,\hat \Lambda_{q-1}) \in \Fqq[x]^q$ such that the congruence equations of \cref{cor:power_keyeq_fx} are satisfied along with the degree bounds on the remainders $\hat B_j$.
By ``minimal'' we will seek a solution such that $\order \hat \Lambda$ is minimal.
The hope is then that $\Lambda = \hat \Lambda$; if that is not the case, we will declare a decoding failure.
In \cref{ssec:power_radius}, we discuss the likelihood of this event occurring in more detail.

Consider first any vector $(\lambda_0,\ldots,\lambda_{q-1}, b_1,\ldots,b_{q\ell}) \in \Fqq[x]^{q(\ell+1)}$ which satisfies the congruences, i.e.
\[
    \sum_{i=0}^{q-1} \lambda_i T_{i,j} \equiv b_j \mod G  , \quad j = 1,\ldots q\ell.
\]
One can quickly see that the space of all such vectors constitutes an $\Fqq[x]$-submodule of $\Fqq[x]^{q(\ell+1)}$.
Furthermore, the rows of the following square matrix $M$ is a basis of this submodule:
\begin{equation}
  \label{eqn:power_matrix}
  M = \left[
    \begin{array}{c|c}
      I_q & T \\\hline
      \vec 0        & G I_{q\ell}
    \end{array}
    \right],
\end{equation}
where $I_m$ is the $m\times m$ identity matrix.
``Solutions'' to the equations are therefore vectors in the $\Fqq[x]$ row-space of $M$ such that the $b_j$ satisfy some degree constraints which are dependent on the $\lambda_i$, and we are seeking a solution where
$\vectifyinv{\lambda_0,\ldots,\lambda_{q-1}}$ has minimal $\order$.
We will handle the latter by finding appropriate weights in the sense of \cref{ssec:mod_weights}, and encode the degree constraints of the $b_j$ as a leading position-constraint on the weighted vector.

Recall the mapping $\Wmap$ of \cref{ssec:mod_weights}.
Let $\vec \eta = (0, q+1,\ldots, (q-1)(q+1))$; then for any $\lambda \in \Ring$ clearly $\WmapO{q,\vec \eta}(\vectify \lambda) = \order \lambda$.

Let now $\mu_j = (q+1)((j-1) \modop q) - m\ceil{j/q} - 1$ for $j=1,\ldots,\ell q$, so the degree constraints for the $b_j$ can be written as
\[
  \deg b_j + \mu_j < \deg(\WmapO{q,\vec \eta}(\lambda_0,\ldots,\lambda_{q-1})).
\]
Some of the $\mu_j$ might be negative, which the method of \cref{ssec:mod_weights} cannot directly handle, so we shift all weights by $\ell m + 1$ to ensure non-negativity.
Therefore letting $\bar\eta_i = \eta_i + \ell m + 1$ and $\bar\mu_j = \mu_j + \ell m + 1$, introduce $\vec w = (\bar\eta_0, \ldots, \bar\eta_{q-1}, \bar\mu_1, \ldots, \bar\mu_{q\ell})$ to realise that the degree constraints can now be written as
\[
  \LP(\WmapH(\lambda_0,\ldots,\lambda_{q-1},b_1,\ldots,b_{q\ell})) \leq q.
\]
Therefore: a minimal solution is a minimal $\WmapH$-weighted vector in the row-space of $M$ among those vectors with leading position in $\{1,\ldots,q\}$.
By the results of \cref{ssec:mod_weights}, we then conclude:
\begin{proposition}
  A minimal solution $(\hat\Lambda_0, \ldots, \hat\Lambda_{q-1})$ to the equations of \cref{cor:power_keyeq_fx} can be found by bringing $\PmapH(M)$ to weak Popov form, and then extracting the row having minimal degree, and in case of a tie least leading position, only among those rows whose leading positions are in $\{ \pi(1), \ldots, \pi(q) \}$.

  The worst-case complexity of computing $M$ and finding the solution is as in \cref{tbl:keyeq_compl} for various choices of the module minimisation algorithm.
\end{proposition}
\begin{proof}
  Only the claim on complexity needs to be discussed further.
  For constructing $M$ we need to compute the sub-matrix $T$, i.e.~for every $t=1,\ldots,\ell$, we need to compute $\matmult{R\T t}\Xi$.
  Computing $R\T 1,\ldots R\T \ell$ requires $\Oapp(\ell n)$ by \cref{lem:gs_interpol}.
  Due to the structure of $\matmult{R\T t}$ and $\Xi$, each element of the matrix product $\matmult{R\T t}\Xi$ requires at most 3 shifts and additions of the elements of $\vectify{R\T t}$, possibly followed by a modulo reduction by $G$, for a total of $\Oapp(n^{4/3})$ operations over $\Fqq$.
  Thus $M$ can be constructed in any of the complexities stated in \cref{tbl:keyeq_compl}.

  For module minimising $\PmapH(M)$, we should estimate $\gamma = \deg M + \max \vec w/q$ as well as $\OD{\PmapH(M)}$.
  For $\gamma$, we have $\deg M = \deg G = n^{2/3}$, while $\max \vec w \leq (q+1)(q-1) + \ell m$.
  As remarked after \cref{thm:power_power}, we can assume $\ell m < n$, and so $\gamma \in O(n^{2/3})$.

  For $\OD{\PmapH(M)} = \rowdeg \PmapH(M) - \deg\deg \PmapH(M)$, we can do better than the generic bound $(\ell+1)q \gamma$: clearly the column permutation performed by $\PmapH$ will not affect the orthogonality defect, and so we should compute the orthogonality defect of $M \diag(x^{w_1},\ldots,x^{w_{q(\ell+1)}})$, where the $w_j$ are the elements of $\vec w$.
  But $M$ is upper triangular, so the determinant is simply the product of the diagonal.
  In the orthogonality defect, therefore only the contribution of the first $q$ rows in the row-degree survives, yielding
  \[
    \OD{\PmapH(M)} \leq q\deg T + \sum_{i=1}^q w_i < q\deg G + q^3 = 2n.
  \]
  Now the entries of \cref{tbl:keyeq_compl} follow from those of \cref{tbl:wpf_complexities}, except that a new entry has been added: the Demand--Driven algorithm from \cite{nielsen_solving_2014} for solving ``asymmetric 2D Pad\'e approximations''.
  This algorithm is derived from the Mulders--Storjohann algorithm, but only applies to matrices coming from such a 2D Pad\'e approximation.
\end{proof}

\begin{table}
  \centering
  \hspace*{-0.2cm}
  \def\Dist{\hspace*{0.8em}}
  \begin{tabular}{@{\Dist}l@{\Dist}l}
    \multicolumn{2}{c}{Complexity for solving the key equations} \\
    \toprule
      Algorithm                                                                    & Field operations in big-$\Oapp$ \\ \midrule
    Mulders--Storjohann \cite{mulders_lattice_2003}                               & $n^{7/3}\ell^2$ \\
    Demand--Driven \cite{nielsen_solving_2014} & $n^{7/3}\ell$ \\
    Alekhnovich \cite{alekhnovich_linear_2005}                                    & $n^{(3+\omega)/3}\ell^\omega$       \\
    GJV \cite{giorgi_complexity_2003} or Zhou--Labahn \cite{zhou_computing_2012}  &  $n^{(2+\omega)/3}\ell^\omega$     \\
    \bottomrule
  \end{tabular}
  \caption
  {Use of $\Oapp$ and $\omega$ as in \cref{tbl:wpf_complexities}.}
  \label{tbl:keyeq_compl}
\end{table}

\subsection{After Having Solved the Key Equation}
\label{ssec:power_after}

We will briefly outline how one can finish decoding once a minimal solution $(\hat \Lambda_0, \ldots, \hat \Lambda_{q-1}, \hat B_1, \ldots, \hat B_{q\ell}) \in \Fqq[x]^{q(\ell+1)}$ to the equations of \cref{cor:power_keyeq_fx} has been found.

Firstly, we apply $\vectifyinv\null$ block-wise to obtain $\ell+1$ elements of $\Ring$: $\hat \Lambda, \hat B_1,\ldots,\hat B_\ell$.
Since the $\Fqq[x]$-vector was found in the row-space of $M$, we know by construction that $\hat \Lambda R\T t \equiv \hat B_t \mod G$ as a congruence over $\Ring$ for $t = 1,\ldots, \ell$.
Therefore, if it is the case that $\Lambda = \hat \Lambda$, then we know by \cref{thm:power_power} that $\hat B_t \equiv \Lambda f^t \mod G$ as a congruence over $\Ring$ for any $t$.
For $t=1$, this congruence can be lifted to equivalence whenever
\begin{IEEEeqnarray*}{rCl+c+rCl}
  \order(\Lambda f) &<& \order(G),  &\textrm{i.e.}&
  |\Errs| &<& n - m - g,
\end{IEEEeqnarray*}
using \cref{lem:power_errloc}.
In that case, we simply need to carry out the division $B_1/\Lambda$ to obtain $f$: we do this by representing the $\Ring$ elements as truncated power series in a local parameter at the place $(0,0)$.
Conversion to and from such power series are discussed in detail in \cref{app:power_series}.
We choose $\phi = x$ as the local parameter, and we can convert $B_1$ and $\Lambda$ into power series in $\phi$ of precision $2q^3$ in time $O(q^4)$ by \vref{prop:power_from_poly}.
Let $\Lambda' = \phi^{-\delta} \Lambda$ and $B_1' = \phi^{-\delta} B_1$ where $\delta$ is the greatest power of $\phi$ that divides $\Lambda$; clearly this will also divide $B_1$ if the correct solution has been found.
Since $0 \neq \Lambda \in \L(|\Errs|P_\infty - \delta(0,0))$ then $\delta < |\Errs|$ which means we obtain the power series of $\Lambda'$ and $B_1'$ to at least precision $q^3$.
Using the extended Euclidean algorithm we can calculate $\Lambda'^\mo \modop x^{q^3}$ in time $\Oapp(q^3)$, and from here $B_1'\Lambda'^{\mo} \equiv f \mod x^{q^3}$ can be calculated in a further $\Oapp(q^3)$ computations.
Finally, converting this truncated power series of $f$ into the standard basis can be done in $O(q^4)$ according to \vref{prop:poly_from_power}.

If we are attempting to decode beyond $n-m-g$, e.g.~for extremely low-rate one-point Hermitian codes (see \cref{prop:power_upper}), then it seems that there is no easy way to obtain $f$ from $\Lambda$ and $(\Lambda f \modop G)$.
An alternative is to find all roots of $\Lambda$ and erase those positions from $\vec r$, and then perform erasure decoding.
We are unaware of a method for doing this in sub-quadratic time, however.

\begin{remark}
  Note that as with the Guruswami--Sudan decoder, we now have a complete decoder which runs in $\Oapp(n^{(2+\omega)/3}\ell^\omega)$, and that the only step of the algorithm with this dominating complexity is module minimisation.
  Thus, again the hidden constant is \emph{exactly} that of the module minimisation algorithm.
  We demonstrate in \cref{sec:simulation} that also in practice the other steps are quite cheap to compute.
\end{remark}

\subsection{Decoding Performance}
\label{ssec:power_radius}

Power decoding is a probabilistic decoding algorithm in the sense that with non-zero probability it might fail for a given received word $\vec r$, i.e.~produce no output.
Indeed, since it can decode beyond half the minimum distance but can return only up to one codeword, this is unavoidable.
However, by simulation it can be observed that the algorithm almost always works up to a very specific bound: this bound is what one could deem ``the decoding radius'' of Power decoding the given code.

This overall behaviour is shared by Power decoding of Reed--Solomon codes \cite{schmidt_syndrome_2010}, but the details turn out to be more involved for one-point Hermitian codes.
We will in this section characterise this behaviour as well as derive the aforementioned bound.
We will repeatedly refer to various events as ``likely'' or ``unlikely'': these are based on statistical observations as well as intuition, but unfortunately we have yet no bounds for most of these probabilities.
It is important future work, but judging from the simpler case of Reed--Solomon codes, where theoretical results have been obtained only for $\ell = 2, 3$ \cite{schmidt_syndrome_2010,zeh_unambiguous_2012,nielsen_power_2014}, it is also rather difficult to obtain such bounds.

For this section we will assume that the sent codeword $\vec c$ is uniquely the closest codeword to $\vec r$; indeed, if there is a different codeword closer or as close to $\vec r$, then it is not surprising that Power decoding with high probability fails or decodes erroneously.
The following result states that when few errors occur, we are guaranteed to succeed:
\begin{proposition}
  \label{prop:power_radius_low}
  The vector $(\Lambda_0,\ldots,\Lambda_{q-1})$ is  a minimal solution to the equations of \cref{cor:power_keyeq_fx} whenever $|\Errs| \leq \frac {\dG-1} 2 - \frac g 2$.
\end{proposition}
\begin{proof}
  Let $(\lambda_0,\ldots,\lambda_{q-1}, \psi_0, \ldots, \psi_{q-1})$ be a minimal solution to the equations for $\ell=1$ while $|\Errs| \leq \frac {n-m} 2 - \frac g 2$, and we will show that $\lambda_i = \gamma\Lambda_i$ for some $\gamma \in \Fqq^*$.
  Since for $\ell > 1$ we impose further restrictions on the solution set, the analogous statement must then be true.
  Let $\lambda = \vectifyinv{\lambda_0, \ldots, \lambda_{q-1}}$ and $\psi = \vectifyinv{\psi_0,\ldots,\psi_{q-1}}$.
  By how the congruence equations and weights for \cref{cor:power_keyeq_fx} were derived, we immediately conclude
  \begin{IEEEeqnarray}{rCl}
    \label{eqn:power_1_congruence}
    \lambda R\T 1 \equiv \psi \mod G
  \end{IEEEeqnarray}
  and $\order \lambda + m + 1 > \order \psi$.
  Thus, $G \mid (\lambda R\T 1 - \psi)$ so by \cref{lem:gs_div} then $\lambda R\T 1 - \psi \in \L(-\sum_{i=1}^n P_i + \infty P_\infty)$.
  Introduce $\hat e = R\T 1 - f \in \Ring$ so $\hat e(P_i) = e_i$ for $i=1,\ldots,n$.
  Clearly $\hat e \in \L(-\sum_{i \notin \Errs} P_i + \infty P_\infty)$, which means
  \begin{equation*}
    \textstyle
    \lambda f - \psi = (\lambda R\T 1 - \psi) - \lambda \hat e \in \L(-\sum_{i \notin \Errs} P_i + h P_\infty),
  \end{equation*}
  where $h$ is an upper bound on $\order( \lambda f - \psi)$: we can choose $h = \order \lambda + m$.
  Now we simply want to show that if $\order \lambda \leq \order \Lambda$ then this Riemann--Roch space is $\{ 0 \}$; for in that case $\lambda f = \psi$, so by \cref{eqn:power_1_congruence} then $G \mid \lambda(R\T 1-f)$, which means $\lambda \in \L(-\sum_{i \in \Errs} P_i + \infty P_\infty)$; but $\Lambda$ has minimal $\order$ of non-zero elements in this Riemann--Roch space, and so $\Lambda = \gamma \lambda$ for some $\gamma \in \Fqq$.

  We have $\L(-\sum_{i \notin \Errs} P_i + h P_\infty) = \{ 0 \}$ at least when the defining divisor has negative degree, and since all $P_i$ and $P_\infty$ are rational, this happens when $n-|\Errs| > h = \order \lambda + m$.
  Now $\order \lambda \leq \order \Lambda \leq |\Errs| + g$ by \cref{lem:power_errloc}.
  Therefore, the divisor is negative at least when
  \begin{IEEEeqnarray*}{+rCl+c+rClCl+x*}
    n-|\Errs| &>& |\Errs| + g + m   & \iff &
    |\Errs| &<& \frac {n - m - g} 2 &=& \frac {\dG - g} 2.
    &
    \ifArxiv\else\IEEEQEDhere\fi
  \end{IEEEeqnarray*}
\end{proof}

We have the following result for when Power decoding does not fail:
\begin{proposition}
  \label{prop:power_closest}
  If Power decoding returns an information polynomial corresponding to the codeword $\hat{\vec c}$, and $\vec c$ is the closest codeword to $\vec r$, then
  \[
    0 \leq \weight{\hat{\vec c} - \vec r} - \weight{\vec c - \vec r} \leq g.
  \]
\end{proposition}
\begin{proof}
  The found solution to the equations of \cref{cor:power_keyeq_fx} is minimal, which means that the corresponding error-locator $\hat \Lambda$ has minimal $\order$ amongst all solutions; in particular $\order \hat \Lambda \leq \order \Lambda$.
  Let $\hat \Errs$ be the error positions corresponding to $\hat{\vec c}$.
  Combining the above with \cref{lem:power_errloc} we get:
  \[
    |\Errs| \leq |\hat \Errs| \leq \order \hat \Lambda \leq \order \Lambda \leq |\Errs|+g
  \]
  and the proposition follows.
\end{proof}
Ideally, we would have hoped that when Power decoding returns a codeword, this is always the closest.
Indeed, that is true for Power decoding of Reed--Solomon codes.
The above states that for one-point Hermitian codes in general, a codeword slightly farther away can actually have the smaller error locator, which will then be found instead.
However, simulations indicate that for random error patterns, the error locator most likely has the maximal order $|\Errs| + g$; most likely, the error locator for either codeword will satisfy this, and so the closest codeword will again have the lowest-order error locator.
The probability of the errors lying such that $\order \Lambda < |\Errs| + g$ was shown to be $1/q$ asymptotically \cite{jensen_performance_1999, hansen_dependent_2001}.

Finally, we will discuss how many errors we should expect Power decoding to be able to cope with.
Recall $M$ of \vref{eqn:power_matrix} whose row space contains all $\Fqq[x]$-vectors satisfying the congruence equations of \cref{cor:power_keyeq_fx}.
The following result puts an upper bound on the $\order$ of the $\lambda$-part of any vector in the row space of $M$:
\begin{proposition}
  \label{prop:power_upper}
  Let $\vec s = \Wmap(\lambda_0,\ldots,\lambda_{q-1},b_1,\ldots,b_{q\ell})$ be the minimal degree vector in the row space of $\Wmap(M)$.
  Then
  \[
    \order\big(\vectifyinv{\lambda_0,\ldots,\lambda_{q-1}}\big) \leq
    \tauPow(\ell) \defeq \tfrac{\ell}{\ell+1} n - \half \ell m - \tfrac{\ell}{\ell+1}.
  \]
\end{proposition}
\begin{proof}
  If $\Wmap(M')$ is a matrix unimodular equivalent with $\Wmap(M)$ and in weak Popov form, then by \cref{prop:wpf} there must be a row $\Wmap(\vec s')$ of $\Wmap(M')$ with $\deg \Wmap(\vec s') = \deg \Wmap(\vec s)$.
  We have
  \begin{IEEEeqnarray*}{rCl}
    \rowdeg(\Wmap(M')) &=& \deg\det(\Wmap(M')) \\
                       &=& \deg\det(\Wmap(M)) \\
                       &=& \textstyle \ell q n + \sum_i\bar \eta_i + \sum_j\bar \mu_j,
  \end{IEEEeqnarray*}
  where $\vec w = (\bar \eta_0,\ldots,\bar \eta_{q-1},\bar \mu_1,\ldots,\bar \mu_{q\ell})$ as specified in \cref{ssec:keyeq_solving}.
  Inserting and simplifying, the right-hand side becomes
  \begin{equation}
    \label{eqn:MpowerRowdeg}
    q \ell n + (\ell+1)\tbinom q 2 + q(\ell+1)\big(g - \half \ell m - \tfrac{\ell}{\ell+1} + (\ell m + 1)  \big).
  \end{equation}
  Clearly $\deg \Wmap(\vec s') \leq \tfrac 1 {q(\ell+1)} \rowdeg(\Wmap(M'))$, but we can do slightly better due to the sparsity of the polynomials in $\Wmap(M')$:
  notice that for any $h$ in $0,\ldots,q-1$, there is exactly one $i$ such that $\eta_i \equiv h \mod q$, and there are exactly $\ell$ indices $j$ such that $\mu_j \equiv h \mod q$.
  Also note that degree of a given row of $\Wmap(M')$ must be congruent modulo $q$ to the weight applied at the leading position.
  Let $\bar h = (\deg (\Wmap(\vec s)) \modop q)$.
  For any $h$ in $0,\ldots,q-1$, since $\Wmap(M')$ is in weak Popov form, there are therefore $\ell+1$ rows whose degree is congruent to $h$ modulo $q$.
  For such a row $\Wmap(\vec m_j)$ we therefore have
  \[
    \deg(\Wmap(\vec m_j)) \geq \deg(\Wmap(\vec s)) + ((h-\bar h) \modop q),
  \]
  where the modulo representative is taken in $0,\ldots,q-1$.
  Summing over all rows we get
  \[
    (\ell+1)\deg(\Wmap(\vec s)) + (\ell+1)\tbinom q 2 \leq \rowdeg \Wmap(M').
  \]
  Finally, by the choice of $\eta_0,\ldots,\eta_{q-1}$, we have
  \[
    \order(\vectifyinv{\lambda_0,\ldots,\lambda_{q-1}}) + \ell m + 1 \leq \deg(\Wmap(\vec s')).
  \]
  Combining these inequalities gives the result.
\end{proof}

The above result therefore states that if $|\Errs| \geq \tauPow(\ell)$ then there are shorter vectors in the row space of $M$ than $\vec \Lambda = (\Lambda, \Lambda f,\ldots, \Lambda f^\ell)$.
These short vectors might not have a leading position within $1, \ldots, q$ as we require from a solution, and $\vec \Lambda$ might still be the shortest vector satisfying this requirement.
However, it seems reasonable to expect that the shortest vector with leading position within $1,\ldots,q$ usually does not have much higher degree than the unconditionally shortest vector: indeed, experiments confirm this, and Power decoding fails almost always when $|\Errs| \geq \tauPow(\ell)$.
See \cref{tbl:sim_failure}.
When it does succeed anyway, this is usually because $\order \Lambda < |\Errs| + g$ as previously discussed.

There is a small caveat to the above discussion: it only holds when $\ell$ is chosen less than or equal to the value which maximises $\tauPow(\ell)$.
For $\ell \rightarrow \infty$ then $\tauPow(\ell) \rightarrow -\infty$.
However, clearly having more key equations is not going to \emph{add} solution vectors, so if, say $(\Lambda, \ldots, \Lambda f^{\ell})$ is the minimal solution choosing some $\ell$, then clearly $(\Lambda, \ldots, \Lambda f^{\hat \ell})$ is the minimal solution when choosing any $\hat \ell \geq \ell$.

Assuming this choice of $\ell$ then whenever $|\Errs| \leq \tauPow(\ell)$, we will most likely succeed and find $\Lambda$.
Unfortunately, we do not have an upper bound on the probability that we fail.
However, our simulations indicate that this probability is low and exponentially quickly decays as $|\Errs|$ fall; see \cref{sec:simulation}.
This is also the case for Power decoding of Reed--Solomon codes, where a proof of these observations is only known for $\ell=2,3$ \cite{schmidt_syndrome_2010,zeh_unambiguous_2012,nielsen_power_2014}.

Recall again from \cref{ssec:power_after} that even when the key equation is solved correctly, we are only able to extract $f$ from $\Lambda$ and $(\Lambda f \modop G)$ when $|\Errs| < n - m - g$.
When $m \ll n$ it is possible that $\tauPow(\ell) > n - m - g$.

\begin{remark}
  For $\ell = 1$, i.e.~minimum distance decoding, then \cref{prop:power_upper} indicates that we will probably succeed when $|\Errs| \leq \frac {n-m-1} 2 = \frac {\dG-1} 2$, while \cref{prop:power_radius_low} only promises success when $|\Errs| \leq \frac {\dG-1-g} 2$.
  This is an interesting, well-known caveat of ``pure'' key equation decoding of AG codes: we are only \emph{assured} decoding success until $g/2$ less than $(\dG-1)/2$, but \emph{almost always}, decoding will succeed all the way until $(\dG-1)/2$.
  The authors are unaware of any work investigating this classical failure probability.
  One can be assured of success all the way to $(\dG-1)/2$ using the majority voting technique of Feng et al.~\cite{feng_simplified_1994}; it is yet unclear whether this technique can be combined with the fast module minimisation and with Power decoding.
\end{remark}

\section{Simulation Results}
\label{sec:simulation}

The proposed algorithms have been functionally implemented in Sage v6.4 \cite{stein_sage_????} and can be downloaded at \url{www.jsrn.dk/code-for-articles}.
The implementation includes basic manipulation of the codes and objects, the fast root-finding and all $\Ya$ conversions.
It does not include either of the fast module minimisation algorithms GJV \cite{giorgi_complexity_2003} or Zhou--Labahn \cite{zhou_efficient_2012}, but instead accomplishes module minimisation using the simpler Mulders--Storjohann algorithm \cite{mulders_lattice_2003}.
The map $\Pmap$ described in \cref{ssec:mod_weights} for handling the weights efficiently has also been implemented.
All parts of our implementation but the module minimisation therefore runs in the asymptotic complexities reported in this paper, though they -- being high-level implementations -- might not have the lowest possible hidden constant.

The implementations allow us to investigate to some degree two concerns which seem difficult to approach analytically: the failure probability of the decoders, and a breakdown of the speed of the various parts of the decoders on concrete parameters.
For the latter, we can -- of course and unfortunately -- say little on the speed of fast module minimisation algorithms.

\subsection{Failure Probability}

We gave in \cref{prop:power_upper} a bound $\tauPow(\ell)$ on how many errors we should expect Power decoding to correct, and conversely, using intuition from linear algebra, we might expect that any number of errors below this will usually be correctable.
This intuition is confirmed by our simulations, which indicate that when $|\Errs| < \tauPow(\ell)$ decoding failure is unlikely, with a probability that quickly decays as $|\Errs|$ falls.
\cref{tbl:sim_failure} summarises simulation results for two different codes.
In the table, for each set of code and decoder parameters, and for each number of errors $\epsilon$, 1000 random codewords were generated and submitted to a random error of Hamming weight exactly $\epsilon$ and attempted decoded.

\def\btauGS{\bar\tau_{\word{GS}}}
It was already observed by Lee and O'Sullivan \cite{lee_list_2009} that Guruswami--Sudan will usually succeed in correcting errors well beyond the guaranteed bound $\tauGS(s,\ell)$ from \cref{sec:gs}, but they gave no description on how much beyond to expect.
Observe that $\tauPow(\ell)$ is exactly $g - \ell/(\ell+1)$ greater than the lower bound on $\tauGS(1,\ell)$ given in \cref{eqn:gs_decoding_radius}.
As can be seen on \cref{tbl:sim_failure}, our simulations indicate that $\tauGS(\ell) + g$ is exactly the bound one should also expect that Guruswami--Sudan will decode up to, when $s=1$.
More generally, there is also an indication that we can expect Guruswami--Sudan to succeed for at least $\tauGS(s,\ell) + g/s$, but more simulations should be carried out to verify this.

For the $q=4$ code, note how the success probability at $\tau + 1$ errors is very close to $q^{-2} = 6.25\%$.
As previously discussed, this is exactly the asymptotic (for $q \rightarrow \infty$) probability that $\order \Lambda < |\Errs| + g$ \cite{jensen_performance_1999,hansen_dependent_2001}, in which case we due to \cref{prop:power_upper} should expect Power decoding to succeed.
The success probability seems better at $\tau+1$ for the $q=5$ code, where $q^{-2} = 4\%$.

For a given $\ell$ and $s=1$, it is a natural question whether there is a correspondence between the cases where Power decoding fails and where Guruswami--Sudan does, for $|\Errs| \leq \tauPow(\ell)$.
It surprised us that we observed no such correspondence: when Power decoding fails, Guruswami--Sudan often succeeds, and vice versa!

\begin{table*}[tb]
  \centering
    \begin{tabular}{l@{\hspace{1em}}rrr@{\hspace{1em}}rrrr}
      \multicolumn{8}{c}{Success probability for the $[64, 10, \geq 49]$ code, with $q=4, m=15$ and $g=6$}                            \\
      \midrule
                          & $\tauGS(s,\ell)$ & $\tauPow(\ell)$ & $\tau$ & $P_s(\tau-2)$ & $P_s(\tau-1)$ & $P_s(\tau)$ & $P_s(\tau+1)$ \\
      GS $(s,\ell)=(1,1)$ & 18               & 24              & 24     & 100\%         & 100\%         & 100\%       & 6.1\%         \\
      Power $\ell=1$      & ---              & 24              & 24     & 100\%         & 100\%         & 100\%       & 6.2\%         \\
      GS $(s,\ell)=(1,2)$ & 21               & 27              & 27     & 100\%         & 100\%         & 93.9\%      & 6.5\%         \\
      Power $\ell=2$      & ---              & 27              & 27     & 100\%         & 100\%         & 94.9\%      & 6.2\%         \\
      GS $(s,\ell)=(2,4)$ & 26               & ---             & 29     & 100\%         & 100\%         & 99.3\%      & 6.5\%         \\[10pt]
      \multicolumn{8}{c}{Success probability for the $[125, 11, \geq 105]$ code, with $q=5, m=20$ and $g=10$}                         \\
      \midrule
                          & $\tauGS(s,\ell)$ & $\tauPow(\ell)$ & $\tau$ & $P_s(\tau-2)$ & $P_s(\tau-1)$ & $P_s(\tau)$ & $P_s(\tau+1)$ \\
      GS $(s,\ell)=(1,2)$ & 53               & 62              & 63     & 100\%         & 99.8\%        & 96.4\%      & 4.5\%         \\
      Power $\ell=2$      & ---              & 62              & 62     & 100\%         & 100\%         & 100\%       & 7.2\%         \\
      GS $(s,\ell)=(1,3)$ & 54               & 63              & 64     & 100\%         & 100\%         & 96.1\%      & 5.1\%         \\
      Power $\ell=3$      & ---              & 63              & 63     & 100\%         & 100\%         & 100\%       & 8.5\%         \\[1pt]
    \end{tabular}
    \caption
    {
      In the above, $\tau = \tauPow(\ell)$ for Power decoding while $\tau = \tauGS(s,\ell) + g/s$ for Guruswami--Sudan.
      $P_s(t)$ denotes the observed probability of decoding success when exactly $t$ errors is added.
    }
    \label{tbl:sim_failure}
\end{table*}

\subsection{Speed}

In our implementation, the running time for both decoders is completely dominated by module minimisation.
Of course, one should recall that our implementations are asymptotically fast in all parts except the module minimisation, where we are using Mulders--Storjohann, so asymptotically, we should expect exactly such a dominance.
However, it is still possible to get an impression on how demanding each part of the decoding algorithms is.
\cref{tbl:sim_speed} shows a breakdown for the time spent on the various parts of the algorithms, using the $[343, 35, \geq 288]$ code having $q = 7, m=55$ and $g=21$.

The reported speeds are the median over 10 trials for each set of parameters.
After module minimisation, Power decoding must perform the division of $\Lambda f$ with $\Lambda$ as described in \cref{ssec:power_after}, while root-finding is performed for Guruswami--Sudan.
``Conversions'' denote time used in converting between the representations of $\Ya$ elements, as described in \cref{app:power_series}.
Precomputation refers to $G$, and various polynomials for Lagrange interpolation as well as for $\Ya$ conversion.
These simulations were run on a laptop with a Core Intel i7-4600U @ 2.1 GHz processor and 8 GB DDR3 1.6 GHz RAM.

We have executed our decoders with various parameters: Power decoding with $\ell=1$ (i.e. classical key equation decoding) and with $\ell=2$, and Guruswami--Sudan with $(s,\ell)=(1,2)$ and $(s,\ell)=(2,4)$.
In all cases, we have run the decoder on the maximal probably decodable number of errors, as discussed in the preceding section.
The received words where decoding failed were discarded from the statistics.

As mentioned, module minimisation completely dominates.
Though we can not draw too final conclusions without an implementation of the asymptotically fast module minimisation algorithm GJV or Zhou--Labahn, even with this algorithm the cost of module minimisation will likely dominate the cost, for even medium sized codes such as this.
In particular, as also predicted by the asymptotic analyses, the cost of conversion between the representations of $\Ya$ elements is highly unlikely to have a significant impact on the total running time.

As is known to be the case for Guruswami--Sudan decoding of Reed-Solomon codes, it seems that also in our case, the root finding is cheaper than the interpolation step.
We can furthermore add that our implementation of Alekhnovich's fast root finding out-performs our implementation of the Roth--Ruckenstein root finding \cite{roth_efficient_2000} already when the $x\deg$ of the input polynomial exceeds 100.

\begin{table*}[tb]
  \centering
    \begin{tabular}{r@{\hspace{1em}}llll}
      \multicolumn{5}{c}{Speed results for the $[343,35, \geq 288]$ code, with $q=7, m=55$ and $g=21$} \\
      \toprule
                              & Power $\ell=1$ & Power $\ell=2$ & GS $(s,\ell)=(1,2)$ & GS $(s,\ell)=(2,4)$ \\
      No. of errors           & 143            & 173            & 173                 & 185                 \\
      Module minimisation     & 2.25 s         & 5.95 s         & 9.36 s              & 177 s               \\
      Division / Root-finding & 0.26 s         & 0.27 s         & 0.36 s              & 0.74 s              \\
      Build matrix            & 0.09 s         & 0.19 s         & 0.15 s              & 0.61 s              \\
      Conversions             & 0.05 s         & 0.05 s         & 0.02 s              & 0.06 s              \\
      Precomputation          & 0.01 s         & 0.01 s         & 0.01 s              & 0.02 s              \\[1pt]
      \midrule
      Total time              & 2.6 s          & 6.4 s          & 9.9 s               & 178 s               \\
    \end{tabular}
    \label{tbl:sim_speed}
\end{table*}

\section{Conclusion}
\label{sec:conclusion}

In this paper, we have demonstrated that decoding of one-point Hermitian codes in sub-quadratic complexity is possible: we describe two decoding algorithms, both of which are able to decode beyond the classical $(\dG-g)/2$ bound.
The main ingredient was to employ recent and deep results in computer algebra for the general problem of $\Fqq[x]$-module minimisation, combined with a new embedding of the original $\Ring$ problem from the function field.

The core of both the Guruswami--Sudan and the Power decoding algorithms seem fairly resilient to the exact function field employed.
We expect in particular that the methods can be extended to one-point codes over any plane Miura-Kamira curve  \cite{sakata_fast_1995} with fairly few changes.
Surprisingly, particular properties of the Hermitian curve, in particular that its equation has only few monomials, were important for attaining sub-quadratic complexity in the auxiliary computations regarding conversion to and from power series; these conversions were necessary for our solutions to the root-finding step in Guruswami--Sudan as well as the post-processing after having solved the key equations in Power decoding.

The decoding algorithms have been functionally implemented in Sage v6.4 \cite{stein_sage_????} and can be downloaded at \url{www.jsrn.dk/code-for-articles}.

\section*{Acknowledgements}

  J.~S.~R.~Nielsen gratefully acknowledges the support of the Digiteo foundation, project IdealCodes.
  Part of this work was also done while he was with Ulm University, and he gratefully acknowledges the support from the German Research Council under grant BO 867/22-1.
  P.~Beelen gratefully acknowledges the support from The Danish Council for Independent Research (Grant No. DFF--4002-00367).

\appendices
\section{Root-finding in \protect{$\Ring[z]$}}
\label{app:root_finding}

For Guruswami--Sudan decoding of one-point Hermitian codes in \cref{sec:gs}, we need to efficiently find all roots of $Q \in \Ring[z]$ whose pole order at $P_\infty$ is less than $m$.
In \cite{beelen_decoding_2008} it was already shown how to solve this problem using the Roth--Ruckenstein algorithm \cite{roth_efficient_2000} for finding $\Fqq[x]$ roots of polynomials in $\Fqq[x][z]$ by adopting a power series view.
We will now show how one can instead apply Alekhnovich's Divide \& Conquer variant \cite{alekhnovich_linear_2005} of the Roth--Ruckenstein algorithm in order to achieve a sub-quadratic complexity in $n$.
The core is a straight-forward power series description of the algorithm of \cite{alekhnovich_linear_2005}, though with a tighter complexity analysis, but for clarity and completeness, we show and prove the complete algorithm.

Consider the rational place $(0,0)$: a local parameter for this place is $\local = x$.
Elements of $\Ring$ have no poles at $(0,0)$, so any $h \in \L(mP_\infty)$ can be written as a power series in $\local$: $h = \sum_{i=0}^\infty h_i \local^i \in \Pring$.
Likewise, we can write $Q = \sum_{t=0}^\ell z^t \sum_{i=0}^\infty q_{t,i} \local^i$.

\begin{lemma}
  \label{lem:root_to_order}
  For any $Q \in \Ring[z]$, consider some $h \in \L(mP_\infty)$ satisfying
  \[
  Q(h) \equiv 0 \mod \local^k,
  \]
  for some integer $k > \orderz m (Q)$ when $Q(h)$ is expanded into a power series in $\phi$.
  Then $Q(h) = 0$.
\end{lemma}
\begin{proof}
  If $Q(h) \neq 0$ then clearly $\order(Q(h)) \leq \orderz m (Q)$.
  Together with the congruence we conclude $Q(h) \in \L\big(\orderz m (Q)P_\infty - k(0,0)\big)$.
  The requirement on $k$ ensures that this Riemann--Roch space contains only 0.
\end{proof}

The strategy is then to iteratively describe all truncated power series $h = \sum_{i=0}^{d_h} h_i \local^i + O(\local^{d_h+1})$ such that $Q(h) \equiv 0 \mod \local^k$ for increasing $k$ until $k > \orderz m (Q)$.
From this set, those that can be extended into functions in $\L(mP_\infty)$ must be unconditional roots of $Q$.
We use the power series conversion detailed in \cref{app:power_series} to convert these roots into functions in the standard basis.
To achieve a quasi-linear dependence on $\orderz m (Q)$, the iterative increments of $k$ are structured in a divide-and-conquer tree.

\def\sols{{\hat \ell}}
\begin{definition}
  For any non-zero $Q \in \Pring[z]$, and some $k \in \ZZ_+$, by \emph{the roots of $Q$ of order $k$}, we will mean the set of $h \in \Pring$ such that $Q(h) \equiv 0 \mod \local^k$.
\end{definition}
The following lemma is an easy extension of \cite[Lemma A.1.1]{alekhnovich_linear_2005}, which in turn was inspired by the analysis of \cite[Section 6]{roth_efficient_2000}:
\begin{lemma}
  \label{lem:root_sets}
  Let $A$ be the roots of $Q$ of order $k$ for any non-zero $Q \in \Pring[z]$ and $k \in \NN_0$,
  Then $A$ can be partitioned into $\sols \leq \deg_z(Q|_{\phi=0})$ many sets $A_1,\ldots,A_\sols$ of the form $A_i = h_i + \local^{d_i} \Pring$ for some $h_i \in \Fqq[\phi]$ and $d_i \in \ZZ_+$.
\end{lemma}
\begin{proof}
  If $Q|_{\phi=0} = 0$ then write $Q = \phi^s \grave Q$ where $\grave Q|_{\phi=0} \neq 0$.
  Then the roots of order $k$ of $Q$ are exactly the roots of order $k-s$ of $\grave Q$.
  Assume therefore that $Q|_{\phi=0} \neq 0$.

  We proceed then by induction on $k$.
  For the base case $k=1$, let $z_1,\ldots,z_\sols$ be the roots of $Q|_{\phi=0} \in \Fqq[z]$.
  Clearly $\sols \leq \deg_z(Q|_{\phi=0})$, and any $h \in A$ will be of the form $h = z_i + O(\phi)$ for one of the $z_i$.

  For the inductive case at $k > 1$, let $z_1,\ldots,z_\sols$ be the roots of $Q|_{\phi=0} \in \Fqq[z]$.
  As before, $\sols \leq \deg_z(Q|_{\phi=0})$, and any $h \in A$ will be of the form $h = z_i + O(\phi)$ for one of the $z_i$.
  Furthermore, let $Q_i = \phi^{-s_i}Q(z_i + \phi z)$ where $s_i \geq 1$ is the greatest integer such that $\phi^{s_i} \mid Q(z_i + \phi z)$.
  It must then be the case that $h = z_i + \phi \grave h$ for some $\grave h \in A_i$, where $A_i$ is the set of roots of $Q_i$ of order $k-s_i$.
  By the induction hypothesis, $A_i$ can be partitioned into $A_{i,1},\ldots,A_{i,\sols_i}$ of the appropriate form, where $\sols_i = \deg_z(Q_i|_{\phi=0})$.
  We can extend each of these $\sols_i$ sets as $A_{i,j}' = z_i + \phi A_{i,j}$ and then $h \in A_{i,j}'$ for some $j$.
  Thus clearly, $A$ can be partitioned into $A_1, \ldots A_L$ of the appropriate form, where $L= \sum_{i=1}^\sols \sols_i$.

  The lemma then follows if we can prove $L \leq \sols$;
  this in turn follows by showing that $\deg_z(Q_i|_{\phi=0}) \leq m_i$ where $m_i$ is the multiplicity of the zero $z_i$ in $Q|_{\phi=0}$.
  We show that by writing $Q = (z-z_i)^{m_i}P_i + \phi \hat Q_i$, where $P_i \in \Fqq[z]$ with $P_i(z_i) \neq 0$ and $\hat Q_i \in \Pring[z]$.
  Then
  \[
    \phi^{s_i} Q_i = (\phi z)^{m_i}P_i(z_i+\phi z) + \phi \hat Q_i(z_i + \phi z).
  \]
  All terms on the right-hand side have $\phi$-degree at least that of the $z$-degree, which means $\deg_z(Q_i|_{\phi=0}) \leq s_i$.
  But $s_i \leq m_i$ since the above right-hand side has the term $(\phi z)^{m_i}P_i(z_i)$, and this can not cancel with any term in $\phi \hat Q_i(z_i + \phi z)$ since these have greater $\phi$-degree than $z$-degree.
\end{proof}

\newcommand{\RootAlg}{\ensuremath{\mathsf{Roots}}}
\begin{algorithm}[t]
  \caption{\RootAlg: root-finding in $\Pring[z]$}
  \label{alg:root}
  \begin{algorithmic}[1]
    \Require{$Q \in \Pring[z], k \in \ZZ_+$}
    \Ensure{The roots of $Q$ of order $k$ in disjoint sets as in \cref{lem:root_sets}, as a set of pairs $(h, d) \in \Fqq[\phi] \times \ZZ_+$ \;}
    \State $Q \ass Q \modop \phi^k$
    \If{$k = 1$}
      \Ifline{$Q = 0$}{\Return $\{(0,0)\}$}
      \State $z_1, \ldots, z_\sols \ass $ $z$-roots of $Q \in \Fqq[z]$\label{line:root_zroots}
      \State \Return $\{(z_i, 1)\}_{i=1}^\sols$
    \Else
      \For{$(h_i, d_i) \in \RootAlg(Q, \ceil{k/2})$}
        \State $\hat Q \ass Q(h_i + \phi^{d_i}z)/\phi^{s_i}$\label{line:root_hatQ}
        \State \quad where $s_i$ is maximal such that $\hat Q \in \Pring[z]$
        \State $\{ (h_{i,j}, d_{i,j}) \}_{i=1}^{\sols_i} \ass \RootAlg(\hat Q, k-s_i)$
      \EndFor
      \State \Return $\{ (h_i + \phi^{d_i} h_{i,j}, d_i + d_{i,j}) \}_{i,j}$\label{line:root_return}
    \EndIf
  \end{algorithmic}
\end{algorithm}

\begin{proposition}
  \cref{alg:root} is correct.
\end{proposition}
\begin{proof}
  We proceed by induction on $k.$
  If $k=1$, clearly the algorithm is correct.
  Now for the inductive step: each root of $Q$ of order $k_\bot = \ceil{k/2}$ will be of the form $h_i + \phi^{d_i}h_i'$ for some $h_i' \in \Pring$ for one of the iterations $(h_i, d_i)$.
  This means $Q(h_i + \phi^{d_i}h_i') \equiv 0 \mod \phi^{k_\bot}$ for any $h_i'$, which is only possible when
  $\phi^{k_\bot} \mid Q(h_i + \phi^{d_i} z)$, implying $s_i \geq k_\bot$ in \cref{line:root_hatQ} for this iteration.

  Now for any $h_i'$, if $h_i + \phi^{d_i}h_i'$ is a root of $Q$ of order $k$, then $\phi^k \mid Q(h_i + \phi^{d_i}z)|_{z=h_i'}$,
  i.e.~$\phi^{k-s_i} \mid \hat Q(z)|_{z=h_i'}$, and so
  $h_i'$ is a root of $\hat Q$ of order $k - s_i$.
  Again by the induction hypothesis $\{(h_{i,j}, d_{i,j})\}$ represents all such roots.
  Therefore, all roots of $Q$ of order $k$ are returned in \cref{line:root_return}.
\end{proof}

\begin{proposition}
  The complexity of \cref{alg:root} is $\Oapp(\ell^2 k)$, where $\ell = \deg_z Q$, assuming $q^2 \in O(k)$.
\end{proposition}
\begin{proof}
  Denote the complexity of the algorithm on input $Q$ with $\deg_z Q = \ell$ and $\deg_z(Q|_{\phi=0}) = \sols$ by $T_\ell(k,\sols)$.
  Note that in none of the recursive calls can we then have $\deg_z Q > \ell$.
  Now, $T_\ell(1, \sols) = O(\sols\costPoly{\sols} \log(q\sols))$, being the complexity of univariate root-finding using e.g.~\cite[Chapter 8.9]{aho_design_1974}, where $\costPoly{n} \in \Oapp(n)$ denotes the complexity of multiplying two polynomials over $\Fqq$ of degree at most $n$ \cite[Theorem 8.23]{von_zur_gathen_modern_2012}.

  For larger $k$, the main loop will have complexity
  \[
    T_\ell(k, \sols) = T_\ell(\ceil{k/2},\sols) + \sum_{i=1}^\sols\big( S_\ell(k-s_i) + T_\ell(k-s_i,m_i) \big),
  \]
  where $S_\ell(k')$ is the cost of computing $\hat Q = Q(h_i + \phi^{d_i})$ when $Q$ is given to precision $\phi^{k'}$, and where the $m_i$ are as in the proof of \cref{lem:root_sets}.
  Recall that $m_1+\ldots+m_\sols \leq \sols \leq \ell$.

  To estimate $S_\ell(k')$, then let $Q = Q_\bot + z^{\ell'}Q_\top$, where $\ell'$ is the greatest power of 2 less than $\ell$, and $\deg Q_\bot < \ell'$.
  Then
  \begin{IEEEeqnarray*}{rCl}
  \phi^{s_i}\hat Q &=& Q(h_i + \phi^{d_i} z)
  \\
  &=& Q_\bot(h_i + \phi^{d_i} z) + (h_i + \phi^{d_i} z)^{\ell'}Q_\top(h_i + \phi^{d_i} z).
  \end{IEEEeqnarray*}
  After precomputation of $(h_i + \phi^{d_i} z)^{2^h}$ for all $h < \log_2(\ell)$, we can compute the product of $(h_i + \phi^{d_i} z)^{\ell'}$ and $Q_\top(h_i + \phi^{d_i} z)$ in complexity $\costPoly{\ell}\costPoly{k'}$.
  Thus we get
  \[
    S_\ell(k') = 2S_{\ell/2}(k') + O(k') + \costPoly{\ell}\costPoly{k'},
  \]
  which my the master theorem \cite{cormen_introduction_2009} has the solution $S_\ell(k') = \Oapp(\ell k')$.

  Back to $T_\ell(k, \sols)$, it is easy to see that the complexity is increasing at least linearly in both $k$ and $\sols$.
  That means that $\sum_{i=1}^\sols T_\ell(k-s_i, m_i) \leq T_\ell(k/2,\sols)$, since the $m_i$ sum to at most $\sols$ and since $s_i \geq k/2$.
  Thus, $T_\ell(k,\sols) \leq \sols S_\ell(k/2) + 2 T_\ell(k/2,\sols)$, which by the master theorem has the solution
  \[
    T_\ell(k,\sols) \in O(\sols S_\ell(k) \log k + kT_\ell(1,\sols)),
  \]
  whence $T_\ell(k, \ell) \in \Oapp(\ell^2 k)$.
\end{proof}

\begin{corollary}
  \label{cor:root_finding_total}
  Given a $Q \in \Ring[z]$ whose coefficients are in the standard basis, we can compute all $f \in \L(mP_\infty)$ in the standard basis such that $Q(f) = 0$ in complexity
  \[
    \Oapp(\ell^2 k + \ell k^2/q^2 + \ell q^4),
  \]
  where $k = \orderz m Q$ and $\ell = \deg_z Q$.
\end{corollary}
\begin{proof}
  By \cref{lem:root_to_order} we need to set $k = \orderz m Q+1$ in \cref{alg:root}, and by \cref{lem:root_sets} we will be returned a list of at most $\ell$ sets of roots.
  The cost of the main algorithm will therefore be $\Oapp(\ell^2k)$.

  Remaining is conversion of input and output.
  We convert $Q$ into an element of $\Pring[z]$ up to precision $k$ using \cref{prop:power_from_poly} for each $\Ring$-coefficient, which we can do in $O(\ell k^2/q^2)$.
  The root sets returned by the root-finding algorithm need to be converted back into the standard basis.
  Note that we do not a priori know the precision of these roots; in particular, whether each have $d_h > m$ so that unique conversion into $\L(mP_\infty)$ is guaranteed by \cref{lem:unique_power_precision}.
  However, even if multiple $\L(mP_\infty)$-element arise from some of the root sets, then each possible element obtained must be an unconditional root of $Q$ by \cref{lem:root_to_order}, and we know that there can be at most $\ell$ such roots in $\Ring \supset \L(mP_\infty)$.
  Thus in total, we will spend $O(\ell q^4)$ on converting the output roots, by \cref{prop:poly_from_power}.
\end{proof}

\section{Power series conversion}
\label{app:power_series}

For both Guruswami--Sudan decoding as well as Power decoding, we need to efficiently convert $\Ring$ elements between the standard basis and truncated power series descriptions.
More precisely, let $\phi$ denote a local parameter for the place $(0,0)$; we will in fact choose $\phi=x$.
We will describe efficient algorithms to do the following: Given a sufficiently long truncated power series development of an element $f \in \L(mP_\infty)$ in $\phi$, compute $f$; and given $f$ compute its truncated power series expansion in $\phi$.
We will show that we can solve both of these problems reasonably efficiently.

Our usual representation of elements in $\Ring$ is an $\Fqq$-combination of the elements in the standard basis:
\[
  S = \{ x^iy^j \mid \, i \ge 0, \, 0\le j \le q-1\}.
\]
Let $S_m \subset S$ denote the elements of $S$ with $\order$ at most $m$, i.e.~$S_m$ is a basis for $\L(mP_\infty)$.

We begin with showing a structural sparsity of $x^iy^j$ monomials when expressed as power series in $\phi$:
\begin{lemma}\label{lem:sparse}
  Let $i$ and $j$ be nonnegative integers.
  In the power series expansion of $x^iy^j$ in $\phi$ up to some precision $N \leq q^3$, there are at most $q$ nonzero coefficients.
  If $j < q$ then for any $N$, there are at most $N/q^2$ nonzero coefficients.
\end{lemma}
\begin{proof}
  First of all, note that a power series development of $y$ in $x$ can be obtained using $y=x^{q+1}-y^q$.
  Iterating this equation, one obtains that
  \begin{IEEEeqnarray}{rCl}
    \label{eqn:y_expanded}
    y &=& \sum_{b=0}^\infty (-1)^b \phi^{(q+1)q^b}.
  \end{IEEEeqnarray}
  Every term in $y^j$ must therefore have a $\phi$-degree of the form $(q+1)\sum_{b = 0}^\infty a_b q^b$ where the $a_b$ are non-negative integers with $\sum_{b=0}^\infty a_b = j$.
  Let $N' = N/(q+1)$ and $r = \floor{\log_q N'}$.
  In the truncation of $y^j$ to precision $N$ the number of terms is then at most the number of tuples $(a_0,\ldots,a_r)$ such that $\sum_{b = 0}^r a_b q^b < N'$ and $a_0 + \ldots + a_r = j$.

  For the lemma's first claim, if $N \leq q^3$ then $N' < q^2$ and $r \leq 1$.
  Since $a_1 > q$ implies $a_0 + a_1 q > N'$ that leaves at most $q$ possible tuples $(j-a_1, a_1)$ for $a_1 = 0,\ldots, \min(q-1, j)$.
 
  For the second claim, assume $j < q$, and therefore also $a_b < q$ for all $b$.
  Note that $\sum_{b = 0}^r a_b q^b$ is then basically some number less than $N'$ written in base $q$, so we are counting how many numbers less than $N'$ have a digit sum exactly $j$.
  We can upper bound that count by counting those numbers with a digital root exactly $j$, which is $\ceil{N'/q}$ of the numbers.
  In total, there must be at most $\ceil{N/(q+1)/q} \leq N/q^2$ non-zero terms in the power series expansion of $y^j$ up to precision $N$.
  Since $x = \phi$, the same holds for $x^iy^j$.
\end{proof}

This immediately implies that it is fast to convert elements from $S_m$ \emph{into} power series:
\begin{proposition}
  \label{prop:power_from_poly}
  Given $f \in \L(M P_\infty)$ described in the basis $S$, we can compute a power series expansion in $\phi$ up to precision $N$ in complexity $O(MN/q^2)$.
\end{proposition}
\begin{proof}
  $f$ is the linear combination of at most $M$ monomials $x^iy^j$, so the power series can be computed by scaling and summing each of these monomial's power series.
  The claim then follows from \cref{lem:sparse}.
\end{proof}

For conversion \emph{from} power series it turns out that a useful stepping stone is a slightly different basis than $S_m$:
\begin{lemma}\label{lem:bases}
Let $m$ be an integer at most $q^3$.
The set
\[
  \hat S_m = \{ x^iy^j \mid qi+(q+1)j \le m, \, 0 \le i \le q, \, j \ge 0\}
\]
is a basis for $\L(mP_\infty)$.
Moreover, any element of $\hat S_m$ can be expressed as a linear combination of at most $q+1$ elements from $S_m$.
\end{lemma}
\begin{proof}
That $\hat S_m$ is a basis for $\L(mP_\infty)$ is clear: writing an element of $\Ring = \Fqq[x,y]$ as a polynomial, the equation of the Hermitian curve $y^q+y=x^{q+1}$ has simply been used to reduce the $x$-degree (where for $S$, the $y$-degree was reduced).

Now let $x^iy^j \in \hat S_m$.
We wish to express it as a linear combination of elements from $S_m$.
It is sufficient to show that $y^j$ with $0\le j \le m/(q+1)$ can be expressed as such a linear combination.
First of all write $j=a+bq$ for unique, nonnegative integers $a$ and $b$ at most $q-1$.
Then we have
\begin{equation}\label{eq:reduce1}
y^j=y^a\left(y^q\right)^b=y^a(x^{q+1}-y)^b.
\end{equation}
If $a+b \le q-1$, this is clearly an expression of $y^j$ as a linear combination of at most $b+1 \leq q$ elements in $S_m$.
If on the other hand $a+b\geq q$, we write $(x^{q+1}-y)^b = p_1 + y^{q-a}p_2$, where $\deg_y p_1 < q-a$ and $\deg_y p_2 \leq a+b-q$, and where both $p_1$ and $y^{q-a}p_2$ are homogeneous in the expressions $x^{q+1}$ and $y$.
Now
\begin{IEEEeqnarray}{rClCl}
  \label{eq:reduce2}
y^j & = & y^a\left(p_1+y^{q-a}p_2\right)
    & = & y^a p_1(x,y)+(x^{q+1}-y)p_2.
\end{IEEEeqnarray}
Note that $\deg_y(y^a p_1) < q$, but that also $\deg_y((x^{q+1}-y)p_2) \leq a+b-q+1 < q$.
Therefore equation \eqref{eq:reduce2} gives the desired expression of $y^j$ as linear combination of elements in $S_m$, and we need estimate only the number of elements in this combination.
But both $p_1$ and $(x^{q+1}-y)p_2$ are homogeneous polynomials in $x^{q+1}$ and $y$ so the number of monomials occurring in each of them is at most their $y$-degree plus one.
This gives a total of at most $(q-a) + (a+b-q+2) = b+2 \leq q+1$ monomials.
\end{proof}

The above lemma shows that we can convert any function $f \in \L(mP_\infty)$ expressed in the basis $\hat S_m$ into the basis $S_m$ in complexity $O(mq) \subset O(q^4)$.\footnote{%
  It is, in fact, easy to show that the reverse conversion can be done with the same complexity, but we will not need that conversion.
}

\begin{lemma}
  \label{lem:unique_power_precision}
Suppose that $f \in \L(mP_\infty)$ with $m < q^3$, and that $m+1$ values $a_i \in \mathbb{F}_{q^2}$ are given such that $f =\sum_{i=0}^{m}a_i \phi^i+O(\phi^{m+1})$.
Then $f$ is determined uniquely.
\end{lemma}
\begin{proof}
Consider $f$ described in the basis $\hat S_m$.
Note that the functions $x^iy^j\in \hat S_m$ have distinct order of vanishing (i.e. valuation $v(\cdot)$) at the place $(0,0)$: indeed $v(x^iy^j)=i+j(q+1)$, and since $0\leq i \leq q$, these quantities will be distinct as $x^iy^j$ runs through $\hat S_m$.
Also, for any $x^iy^j \in \hat S_m$ we have $v(x^i y^j) \leq m < q^3$.
The coefficients $a_i$ therefore uniquely determine a linear combination $g=\sum_{i,j} b_{i,j} x^iy^j \in \L(mP_\infty)$ of elements in $\hat S_m$ such that $v(f-g)>m$.
This implies that $f-g \in \L(mP_\infty-(m+1)(0,0))$.
However, since that divisor clearly has negative degree, the Riemann-Roch space must be $\{0\}$, implying $f=g$ as desired.
\end{proof}

Note that when computing the linear combination $g=\sum_{i,j} b_{i,j} x^iy^j$ in the above proof, we are essentially using back-substitution: one finds $x^iy^j \in \hat S_m$ and $c \in \mathbb{F}_{q^2}$ such that $v(f-cx^iy^j)>v(f)$, i.e., one eliminates the lowest order term in the approximate power series development of $f$.
Then one updates $f$ to $f-cx^iy^j$ (as well as the corresponding truncated power series) and iterates this process till all coefficients in the truncated power series of $f$ are eliminated.
By \cref{lem:sparse} an update can be performed in $O(q)$.
The total construction of $g$ therefore can be done in $O(q^4)$.
If one ends in the situation that a coefficient in the (updated) truncated power series cannot be eliminated by adding a multiple of a power series development of an element from $\hat S_m$, then the conclusion is that for no $f \in \L(mP_\infty)$ it holds that $f =\sum_{i=0}^{m}a_i \phi^i+O(\phi^{m+1})$.
Otherwise, one ends up with $f$ described in the basis $\hat S_m$, and one can then convert into the basis $S_m$.
All in all, we have shown the following:

\begin{proposition}
  \label{prop:poly_from_power}
Let $m< q^3$.
Given a truncated power series development to precision $m+1$ for an element $f$, one can determine whether or not $f \in \L(mP_\infty)$, and in the affirmative case express $f$ in the basis $S_m$ in complexity $O(q^4)$.
\end{proposition}

\bibliographystyle{ieeetr}
\bibliography{bibtex}

\end{document}